\documentclass[11pt]{article}
\usepackage{fullpage}           
\usepackage{amsfonts}           
\usepackage{amsthm}             
\usepackage{amsmath}
\usepackage{dsfont}
\usepackage{amssymb}
\usepackage{xspace}             
\usepackage{graphicx}            
\usepackage{adjustbox}          
\usepackage{multicol}            
\usepackage{url}
\usepackage{enumerate}  
\usepackage{subcaption}
\usepackage[usenames]{xcolor} 

\usepackage[letterpaper,margin=1in]{geometry}
\usepackage{hyperref}
\usepackage{mathabx}
\usepackage{sidecap}

\usepackage[ruled,linesnumbered,vlined]{algorithm2e}
\colorlet{darkgreen}{green!45!black}

\allowdisplaybreaks

\newtheorem{theorem}{Theorem}[section]
\newtheorem{definition}[theorem]{Definition}

\newtheorem{conjecture}[theorem]{Conjecture}
\newtheorem{lemma}[theorem]{Lemma}

\newtheorem{observation}[theorem]{Observation}
\newtheorem{claim}[theorem]{Claim}

\newtheorem{corollary}[theorem]{Corollary}

\newcommand{\calV}{\mathcal{V}}

\newcommand{\calS}{\mathcal{S}}
\newcommand{\R}{\mathbb{R}}

\newcommand{\N}{\mathbb{N}}

\newcommand{\dist}{\text{dist}}
\newcommand{\eps}{\varepsilon}
\newcommand{\opt}{\text{OPT}\xspace}

\newcommand{\calC}{\mathcal{C}}

\newcommand{\calT}{\mathcal{T}}

\newcommand{\ind}{\mathbb{I}}
\newcommand{\Inf}{\mathsf{Inf}}
\newcommand{\Ind}{\mathbf{1}}
\newcommand{\Valug}{\mathsf{Val_{UG}}}
\newcommand{\Optug}{\mathsf{Opt_{UG}}}

\newcommand{\E}{\mathbb{E}}

\newcommand{\mvd}{\textsc{Metric Violation Distance}\xspace}
\newcommand{\maxmvd}{\textsc{Max-Metric Violation Distance}\xspace}
\newcommand{\umvd}{\textsc{Ultrametric Violation 
Distance}\xspace}
\newcommand{\maxumvd}{\textsc{Max-Ultrametric Violation Distance}\xspace}
\newcommand{\kmedian}{\textsc{$k$-median}\xspace}
\newcommand{\correlationclustering}{\textsc{Correlation Clustering}\xspace}

\newcommand{\gz}{\geq 0}
\newcommand{\nct}{\binom{[n]}{2}}

\newcommand{\tsp}{\textsc{Traveling Salesman Problem}\xspace}

\newcommand*{\inlineequation}[2][]{%
  \begingroup
    \refstepcounter{equation}%
    \ifx\\#1\\%
    \else
      \label{#1}%
    \fi
    \relpenalty=10000 %
    \binoppenalty=10000 %
    \ensuremath{%
      #2%
    }%
    ~\@eqnnum
  \endgroup
}

\usepackage{authblk}
  \def\rem#1{{\marginpar{\raggedright\scriptsize #1}}}

\ifdefined\DEBUG
 \newcommand{\vca}[1]{\rem{\textcolor{blue}{$\bullet$ #1}}}
\newcommand{\vi}[1]{(V: \textcolor{blue}{#1})}
\newcommand{\el}[1]{\rem{\textcolor{gray}{$\bullet$ #1}}}
\newcommand{\eui}[1]{\textcolor{gray}{#1}}
\newcommand{\ar}[1]{\rem{\textcolor{red}{$\bullet$ #1}}}
\else
 \newcommand{\vi}[1]{}
 \newcommand{\vca}[1]{}
 \newcommand{\eui}[1]{}
 \newcommand{\el}[1]{}
 \newcommand{\ar}[1]{}
 \fi

\title{Fitting Metrics and Ultrametrics with Minimum Disagreements}
\author{Vincent Cohen-Addad \thanks{Google Z\"urich, Switzerland.} \and Chenglin Fan \thanks{Sorbonne Université, Paris, France.} \and Euiwoong Lee\thanks{University of Michigan, Ann Arbor, Michigan, USA. Partially supported by a gift from Google.} \and Arnaud de Mesmay\thanks{LIGM, CNRS, Univ. Gustave Eiffel, ESIEE Paris, F-77454 Marne-la-Vallée, France.}}
\date{}

\begin{document}

\maketitle
\begin{abstract}
        Given $x \in (\R_{\geq 0})^{\binom{[n]}{2}}$ recording pairwise distances, the \mvd problem asks to compute the $\ell_0$ distance between $x$ and the metric cone; i.e., modify the minimum number of entries of $x$ to make it a metric. Due to its large number of applications in various data analysis and optimization tasks, this problem has been actively studied recently.
    
    We present an $O(\log n)$-approximation algorithm for \mvd, exponentially improving the previous best approximation ratio of $O(OPT^{1/3})$ of Fan, Raichel, and Van Buskirk [\textit{SODA}, 2018]. Furthermore, a major strength of our algorithm is its simplicity and running time. 
    We also study the related problem of \umvd, where the goal is to compute the $\ell_0$ distance to the cone of ultrametrics, and achieve a constant factor approximation algorithm.
    The  \umvd problem can be regarded as an extension of the problem of fitting ultrametrics studied by Ailon and Charikar [\textit{SIAM J. Computing}, 2011] and by  Cohen-Addad, Das,  Kipouridis,   Parotsidis, and  Thorup [\textit{FOCS}, 2021]
    from $\ell_1$ norm to $\ell_0$ norm. We show that this problem can be favorably interpreted as an instance of 
    \textsc{Correlation Clustering} with an additional hierarchical structure, which we solve using a new $O(1)$-approximation algorithm for correlation clustering
    that has the structural property that it outputs a refinement of the optimum clusters. An algorithm satisfying such a property can be considered of independent interest.
    We also provide an $O(\log n \log \log n)$ approximation algorithm for weighted instances. 
    Finally, we investigate the complementary version of these problems where one aims at choosing a maximum number of entries of $x$ forming an (ultra-)metric. In stark contrast with the minimization versions, we prove that these maximization versions are hard to approximate within any constant factor assuming 
    the Unique Games Conjecture.

\end{abstract}
\section{Introduction}
Numerous tasks arising in data analysis and optimization deal with data given as {\em distances} between pairs of objects. The triangle inequality, which simply states that ``the distance from $i$ to $j$ cannot be strictly greater than the distance from $i$ to $k$ plus the distance from $k$ to $i$,'' is arguably the most natural constraint one can expect from the set of pairwise distances. Formally in this paper, $x \in (\R_{\geq 0})^{\binom{[n]}{2}}$ is called a {\em metric} if for every distinct $i, j, k \in [n]$, $x(i, j) \leq x(i, k) + x(j, k)$. Note that we allow distances to be zero. The metric structure is desirable in most machine learning and data analysis tasks where the input
models dissimilarity between objects, e.g.: nearest neighbors, metric learning, and clustering.


The metric structure also allows for the existence of efficient algorithms with provable theoretical guarantees. For example, many central optimization problems in clustering and network design including \kmedian and \tsp have no polynomial time algorithms with finite approximation ratios without the assumption that the given distances form a metric, but they admit constant factor approximation algorithms with the assumption.
Furthermore, the metric structure has also helped develop faster algorithms (e.g., \cite{Elkan2003} uses the triangle inequality to accelerate \textsc{$k$-means}).

However, when the distance data $x \in \R_{\geq 0}^{\binom{[n]}{2}}$ is obtained by measurements, it may be the case that the data does not form a metric due to noise of the measurement, missing data, and other corruptions. In this scenario, it is natural to find the metric $y$ {\em closest to} the given data. Assuming that the uncorrupted data should form a metric, $y$ can be considered as the {\em denoised} version of $x$ on which one can perform various tasks using the metric structure. 
 For example, the
 mPAM matrices~\cite{Dayhoff78chapter22}, which represent a certain measure of dissimilarity in protein sequencing, tend  not to be distance matrices. However, the  query can be  accelerated  in   biological databases once  a metric-based data indexing scheme can be constructed.


The measure of {\em closeness} we study in this paper is the $\ell_0$ distance $\| x - y \|_0 = | \{ (i, j) \in \binom{[n]}{2} : x(i, j) \neq y(i, j) \} |$. 
This choice seems natural in practice because often pairwise distances are obtained by different (human) classifiers, and while most classifiers will do a good job, if one classifier makes an error, it may make a large amount; if $x$ originally forms a metric but one entry $x(i, j)$ is increased to a very large number, the $\ell_0$ objective will return the original metric where convex objectives like $\ell_1$ or $\ell_2$ may likely change every distance. 


\paragraph*{Metric Violation Distance.} Hence we study the following \mvd problem: Given $x \in \R_{\geq 0}^{\binom{[n]}{2}}$, find $y \in M_n$ that minimizes $\|x - y\|_0$, where $M_n \subseteq \R^{\binom{n}{2}}$ is the set of all metrics on $n$ points.   
Brickell, Dhillon, Sra, and Tropp~\cite{brickell2008metric} first formulated the problem with $\ell_p$ objectives for $1 \leq p \leq \infty$, where the problem can be solved exactly via linear or convex programming. The $\ell_0$ version has been recently introduced and actively studied~\cite{gilbert2017if, fan2018metric, fan2020generalized}. 
Motivated by resolving inconsistencies in biological measurements, the maximization variant of the $\ell_0$ objective (i.e., maximize $\binom{n}{2} - \| x - y \|_0$) was also studied~\cite{duggal2013resolving}. 

The best approximation ratio for \mvd is  $O(OPT^{1/3})$, where $OPT$ is the minimum $\ell_0$ distance which can be as big as $\Omega(n^2)$, and the corresponding algorithm runs in time $O(n^6)$~\cite{fan2018metric}. The best hardness of approximation ratio is $2$ assuming the Unique Games Conjecture~\cite{fan2018metric}.

Our first result is the following $O(\log n)$ approximation algorithm for \mvd.
This is an exponential improvement over the previous best ratio of $O(n^{2/3})$ which 
is also significantly faster. 

\begin{theorem}
There exists a randomized $O(\log n)$-approximation algorithm for \mvd that runs 
in time $O(n^3)$. 
\label{thm:mvd}
\end{theorem}

Besides the improved approximation ratio, the strength of our algorithm is its simplicity. Note that the $O(n^3)$ complexity is actually not cubic since the input has size $\Omega(n^2)$. Furthermore, since in particular the algorithm can be used to verify whether the input data is already a metric (deciding whether the objective is 0 or not), the $O(n^3)$ complexity is tight by a result of Williams and Williams~\cite{williams2018subcubic} (assuming a standard fine-grained complexity conjecture). 

Our algorithm is a simple pivot-based algorithm (see Algorithm~\ref{algo:mvd}) which, at each iteration, chooses a random pivot $i\in[n]$ and minimally changes the values of $\{x(j,k)\}$ so as to ensure that all the triangle inequalities involving $i$ are satisfied. Then it recursively solves the problem on the smaller instance $[n] \setminus \{i\}$. While this algorithm is very similar in spirit to the pivot-based algorithms first introduced by Ailon, Charikar and Newman~\cite{ailon2008aggregating} for various ranking and clustering problems, our main innovation lies in the analysis: compared to the aforementioned problems, here the value of some distances might be modified multiple times throughout the algorithm, and thus the analysis requires different techniques, leading to a slightly worse approximation ratio (logarithmic vs. constant). We also provide an example showing that our analysis is tight.




\paragraph*{Ultrametric Violation Distance.} It is also natural to ask the same question for special classes of distances. One special class we study in this paper is the class of {\em ultrametrics}; $x \in \R_{\geq 0}^{\binom{n}{2}}$ is called an {\em ultrametric} if every distinct $i, j, k \in [n]$, $x(i, j) \leq \max(x(i, k), x(k, j))$. Equivalently, $y \in \R^{\binom{[n]}{2}}_{\gz}$ is an ultrametric if and only if there exists a tree such that (1) each edge has a distance, (2) the distance from the root to every leaf is the same, (3) the leaves are labeled by $[n]$, and $y$ is given by the tree distances between the leaves. Let $U_n$ be the set of all ultrametrics on $n$ points, and \umvd be the problem, where given $x$, the goal is to find $y \in U_n$ such that $\| y - x \|_0$ is minimized. 
Since ultrametrics are closely related to {\em hierarchical clustering}, they lie
at the heart of a large number of unsupervised learning approaches, such 
as for example the classic linkage algorithms, see the survey of~\cite{carlsson2010characterization}. There has thus been several results recently
on computing ultrametric embeddings of input points minimizing various objectives
(\cite{MW17,RP16,CC17,Cohen-AddadKM17,CCN19,CCN18,Cohen-AddadKMM19,AlonAV20,ChatziafratisYL20,Balcan08,ChamiGCR20,cochez2015twister,AbboudCH19, sonthalia2020tree,DBLP:conf/aistats/VainsteinCCRMA21}).

\ar{Reviewer 2: They mentioned that the ultrametric case is interesting due to the fact that in general ultrametrics and strongly linked to hierarchical clusterings. However, I do feel this specific problem (of violation distance) should be better motivated.}


We are not aware of any previous work on this \umvd problem with a $\ell_0$ objective, except that its APX-hardness can be easily deduced from the APX-hardness of \correlationclustering~\cite{charikar2005clustering}. Ailon and Charikar~\cite{ailon2011fitting} studied the $\ell_1$ objective
version of the problem and gave a $\min(T+2, O(\log n \log \log n))$-approximation where $T$ is the height of the tree, 
improving on the $O(T \log n)$-approximation of Harb, Kannan, and McGregor~\cite{Harb05approximatingthe}.
Recently Cohen-Addad, Das,  Kipouridis,   Parotsidis, and  Thorup~\cite{cohenaddad2021fitting} improved the approximation ratio to $O(1)$ hence giving the first constant factor for the unweighted version
of the problem in its full generality.
There is also a large body of work on minimizing the maximum distortion, namely for the 
$\ell_{\infty}$ objective, for which a 3-approximation is known since the 90s~\cite{DBLP:journals/siamcomp/AgarwalaBFPT99} and 
recent progress has been made when the input lies in a high-dimensional Euclidean space~\cite{Cohen-AddadJL21,Cohen-AddadSL20}.
Sidiropoulos, Wang and Wang~\cite{DBLP:conf/soda/SidiropoulosWW17} studied the {\em outlier deletion} version of the problem where the goal is to find the smallest $S \subseteq [n]$ such that $x$ induced by $[n] \setminus S$ becomes a ultrametric (among others), and gave a $3$-approximation algorithm.

Our pivot-based approach for \mvd is robust, and as we will see, a minor modification of 
Theorem~\ref{thm:mvd} (see Algorithm~\ref{algo:umvd}) yields a simple $O(\log n)$-approximation algorithm for \umvd as well. Furthermore, for this ultrametric variant, we show that our analysis is tight by proving that Algorithm~\ref{algo:umvd} cannot give better than an $O(\log n)$-approximation {\em not only for random choices of pivots, but for any possible choice of a sequence of pivots}.  

\begin{theorem}
For infinitely many $n$, 
there is an instance $x \in \R_{\gz}^{\nct}$ such that for any choices of pivots, Algorithm~\ref{algo:umvd} outputs an
$\Omega(\log n)$-approximate solution. 
\label{thm:mvd-lower}
\end{theorem}

Our next result shows how to overcome this logarithmic barrier using different techniques.  By seeing the problem as a hierarchical variant of \correlationclustering, we are able to develop a new constant-factor approximation for \correlationclustering, and leverage it to obtain an $O(1)$-approximation algorithm ofr \umvd. 

\begin{theorem}
There exists a polynomial time deterministic $O(1)$-approximation algorithm for \umvd.
\label{thm:umvd_noweight}
\end{theorem}
\vca{make running time explicit?}

\paragraph*{Violation Distance for Weighted Instances.} For both \mvd and \umvd, it is also natural to study the extension of the problems to {\em weighted instances}: the input consists of distances $x \in \R_{\gz}^{\binom{[n]}{2}}$ and weights 
$w \in \R_{\gz}^{\binom{[n]}{2}}$, and the goal is to find a metric or ultrametric $y \in \R_{\gz}^{\binom{[n]}{2}}$ to minimize $\sum_{(i,j)}w(i,j)\cdot\ind(x(i,j)\neq y(i,j))$.

These more general versions may be harder to approximate than the unweighted version; there exist polynomial-time approximation-preserving reductions from both \textsc{Multicut} and \textsc{Length-Bounded Cut} to \mvd on weighted instances, see Fan, Gilbert,  Raichel,  Sonthalia,  and  Van Buskirk~\cite{fan2020generalized}.
The best approximation ratios for them are $O(\log n)$~\cite{GargVY96} and $O(n^{2/3})$~\cite{baier2010length} respectively, and both problems are hard to approximate within any constant factor assuming the Unique Games Conjecture~\cite{ChawlaKKRS06, l-ihcifp-17}. The standard LP relaxation for \textsc{Length-Bounded Cut} has an integrality gap of $\Omega(n^{2/3})$, so approximating the \mvd problem on weighted instances with a subpolynomial approximation ratio is already a major challenge.

For \umvd, the reduction from \textsc{Multicut} still holds, implying that it seems unlikely to have an $O(1)$-approximation for weighted graphs. However, we adapt the algorithm of Ailon and Charikar~\cite{ailon2011fitting} to design an $O(\log n \log \log n)$-approximation algorithm even for weighted instances.


\begin{theorem}
There exists a polynomial time $O(\log n \log \log n)$-approximation algorithm for \umvd on weighted instances. 
\end{theorem}

\paragraph*{Maximization version and hardness.}

Finally, we also investigate the complexity of the complementary (maximization) versions of \mvd and \umvd problems, which we call \maxmvd and \maxumvd. Here the problem consists of a set of vertices $[n]$ and a set of pairs $(i ,j) \in \binom{n}{2}$ with distance $x(i, j)$ and weight $w(i, j)$, and the goal is to compute an (ultra)metric $\dist : \binom{n}{2} \to \R$ so as to maximize the total weight of pairs $(i, j)$ with $x(i, j) = \dist(i, j)$. This maximization version has been studied in the context of molecular biology~\cite{duggal2013resolving}, where heuristic approximations algorithms have been provided. We prove that these problems are Unique Games-hard to approximate within any constant factor.

\begin{theorem}
Assuming the Unique Games Conjecture, it is NP-hard to approximate \maxmvd and \maxumvd within any constant factor.
\label{thm:max-hard-complete}
\end{theorem}

This hardness result also holds for the instances that are unweighted and complete (see Corollary~\ref{cor:complete} and~\ref{cor:mvd}), and therefore stands in contrast with the minimization version of \umvd, for which our Theorem~\ref{thm:umvd_noweight} provides a constant-factor approximation. For these maximization problems, we are not aware of any non-trivial approximation algorithms. Since these problems can be reformulated as choosing a maximum number of edges that avoid specific unbalanced cycles (see the summary of techniques below), a similarly-looking problem is the one of finding a subgraph with a maximum number of edges that does not contain a $k$-cycle for some fixed $k$. That problem has been studied by Kortsarz, Langberg and Zutov~\cite{kortsarz2010approximating}, and features similar huge gaps between the best known approximation algorithms and the lower bounds.

\subsection*{Challenges and New Techniques.} 

\begin{figure}
    \centering
    \def\svgwidth{12cm}
    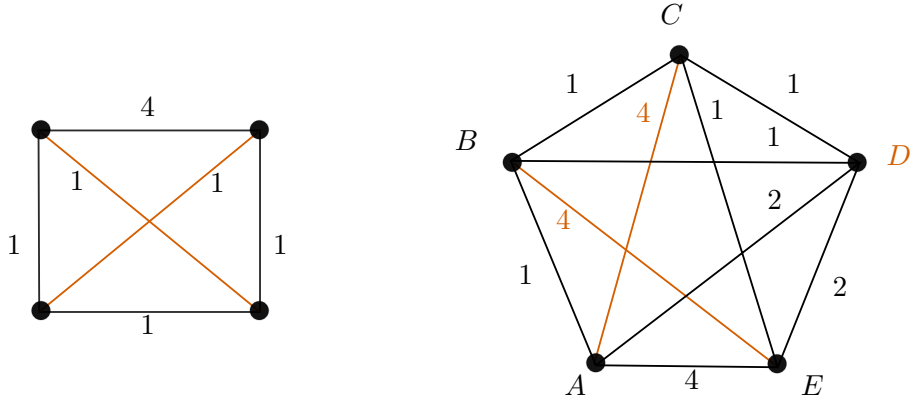
    \caption{Left: The two colored edges hit all the violated triples, but one cannot assign them new distance values without creating new violated triangles. Right: If we first pivot at D and remove the corresponding offending edges BE and AC, there are no violated triples anymore but no way to assign new distance values to AC and BE without creating new violated triangles.}
    \label{F:badexamples}
\end{figure}

\paragraph*{Pivot-based algorithms for (Ultra)Metric Violation Distance.} From the outset, it might look like \mvd can be phrased as a simple hitting set problem for a family of violated triples, namely, the triples of distances that do not satisfy the triangle inequality. This would suggest an easy constant-factor approximation, by simply approximating the corresponding hitting set instance. However, it might happen that even though one finds a family of edges that hits all the bad triples, there is no way of changing the value of the distances on this family of edges without creating new bad triples. Such a behavior is pictured on the left of Figure~\ref{F:badexamples} where we can remove two edges to hit all the violated triples, yet there is no way to choose corrected distances for them without inducing new metric violations.  The culprit here is easily found: there is a violated $4$-tuple, i.e., a cycle (when viewing the instance as a complete graph) where one distance is longer than the sum of the other distances on the cycle. This thus makes the \mvd problem considerably harder than a hitting set problem on bounded size cycles.

It was shown in~\cite{fan2018metric} that hitting all of the violated cycles is equivalent to the \mvd problem, but there are potentially exponentially many of those (i.e.: $2^{\Omega(\sqrt{n})}$ bad cycles of length $\sqrt n$, etc.), and thus no natural approximation algorithm exists. In~\cite{fan2018metric}, the authors provide an $O(OPT^{1/3})$ approximation algorithm by focusing primarily on bad cycles of length at most $6$, which in particular requires enumerating those and has a high complexity cost.

In contrast, our approach in Algorithm~\ref{algo:mvd} is to nevertheless focus on violated triples by using a pivot-based algorithm: We choose uniformly at random a pivot, freeze the values of its incident distances, and repair the potential violated triples incident to the pivot by changing the value of the third distance in such a triple. Then we recurse, i.e., pivot on another vertex\footnote{Throughout this paper, we use the words ``vertex'' and ``point'' interchangeably, since our algorithms are inspired by graph algorithms.}
until the entire metric has been repaired. The entire pseudocode is displayed in Algorithm~\ref{algo:mvd}.

\begin{algorithm}[t]
    {\bf Input:} $n \in \N$, $x \in \R_{\gz}^{\nct}$, 
    {\bf Goal:} modify $x$ to a metric.\\
    \If{$n \leq 2$}{{\bf return}}
        
    Choose the pivot $i \in [n]$ uniformly at random\;
    \For{$j \neq k \in [n] \setminus \{ i \}$}{
    \If{$x(j, k) > x(i, j) + x(i, k)$}{
        $x(j, k) = x(i, j) + x(i, k)$;
    }
    \If{$x(j, k) < |x(i, j) - x(i, k)|$}{
    $x(j, k) = |x(i, j) - x(i, k)|$;
    }
    }
    \textsc{MVD-Pivot}$(n - 1, x|_{\binom{[n]\setminus i}{2}})$ \tcp*{$x|_{\binom{[n]\setminus i}{2}}$ is $x$ restricted to $\{ (j, k) \in \binom{[n] \setminus \{ i \} }{2} \}$}
    {\bf return}
 \caption{\textsc{MVD-Pivot}$(n, x)$}
 \label{algo:mvd}
\end{algorithm}

The crux of the argument is to define the appropriate reparation: as the example on the right of Figure~\ref{F:badexamples} shows, simply removing the offending edges from our graph is not enough, as it might not repair violated $k$-tuples for $k>3$. Our choice in Algorithm~\ref{algo:mvd} is to minimally modify the distance of the offending edge so as to satisfy the triangle inequality (thus making it an equality). As we will show in Section~2, this choice yields very desirable properties (Lemma~\ref{L:goodstaysgood} and~\ref{L:betterlemma}). However, an offending edge may be modified multiple times during the course of our algorithm, which gives it a different flavor compared to existing pivot algorithms. We bound this number of changes by a logarithmic quantity in Theorem~\ref{thm:mvd}, explain how it connects to the approximation ratio, and we also provide an example showing that this is tight.\ar{Feel free to reword, this is not very punchy.}
An interesting phenomenon here is that Ailon and Charikar~\cite{ailon2011fitting} in their work on fitting ultrametrics
under the $\ell_1$ objective only obtained an $O(T)$-approximation 
algorithm using pivot-based approach, where $T$ is the number of different input distances. Here, we manage to use the structure of the problem together with a refined pivot algorithm so as to obtain
an $O(\log n)$-approximation algorithm.

An almost identical algorithm and analysis also provides a logarithmic approximation ratio for \umvd. Furthermore, in the ultrametric setting, a very symmetric construction based on hypercubes allows us to pinpoint precisely the limitation of pivot-based techniques: we build an instance of \umvd where \emph{any} choice of a sequence of pivots yields a solution that differs from the optimal by a logarithmic factor. In order to overcome this gap, this leads us to designing a completely different approach to obtain a better approximation factor for \umvd.

\paragraph*{Ultrametric Violation Distance via Correlation Clustering.}

\umvd offers additional structure compared to \mvd, as an ultrametric can favorably be interpreted as a hierarchical clustering:  for any $m$, the sets of points at distance smaller than $m$ form
natural equivalence classes and thus can be thought of as a clustering of the input. 

This allows us to define a natural correlation clustering instance for finding groups of points that are distance less than $m$ and groups of points that are at distance at least $m$: place a ``$+$ edge" for any 
edge that has value strictly less than $m$ and a ``$-$ edge" for any edge that has value at least $m$. This defines an instance of correlation clustering 
on complete unweighted graphs; indeed, the goal of correlation clustering
is to partition the input into groups such that the total number of ``$-$'' edges whose extremities are in the same group plus the total number of $+$ edges whose extremities are in different groups is minimized. 

Thus, for a given distance $m$, the optimum solution for the whole \umvd problem
induces a natural partitioning into groups (points at distance at most $m$) that implies a solution to the correlation clustering instance defined at level $m$. The cost of the two solutions are similar.
On the one hand, edges of length at least $m$ whose extremities are both in the same group will be set to a distance smaller than $m$ and so induce a cost of 1, and since they are ``$-$'' edges in the correlation clustering instance, they also induce a cost of 1. 
On the other hand, edges of length less than $m$ whose extremities
are in different groups, are at distance at least $m$ in the \umvd solution and thus incur a cost of 1, and are ``$+$'' edges across clusters in our correlation clustering instance and also incur a cost of 1.

Thus, the \umvd problem can be cast as a sequence of correlation clustering instances where the goal is to find nested correlation clustering solutions where one aims at minimizing the number of violated edges (an edge being violated if it is a ``$+$'' edge at some level and
its extremities are in different clusters at that level, or if it is a ``$-$'' edge at some level and its extremities are in the same cluster). 

To solve instances of the above problem resulting from \umvd instances, 
we use a new $O(1)$-approximation algorithm for correlation clustering.
It is worth noting that one cannot simply compute an $O(1)$-approximation to correlation clustering for the largest distance and recurse in each cluster. This is because ``$-$'' edges at the same distance are also minus edges at smaller distances and one may risk charging a ``$-$'' edge
that is paid by the optimum solution at each distance until the extremities
are separated by the algorithm.
To avoid the above, we describe a new algorithm for correlation clustering that satisfies the property
that the clusters it outputs do not mix the ``important'' clusters of the optimum solution and so essentially produce clusterings that are refinement of an (approximately) optimal solution. We can
thus compute the ultrametric top-down distance by distance. \ar{Reviewer 3 maybe has a point that we have no idea what the algorithm does after reading this.}

We also provide an $O(\log n \log \log n)$-approximation algorithm for \umvd on weighted instances. It relies on the rounding of a natural linear programming relaxation for the underlying hierarchical clustering problem. The algorithm and its Seymour-style~\cite{seymour1995packing} analysis are heavily inspired by a similar algorithm of Ailon and Charikar~\cite{ailon2011fitting} to solve the $\ell_1$ variant of the problem.

\paragraph*{Hardness of the maximization versions.} The previous algorithms (except the very last one) all rely extensively on the fact the input graph is complete. In contrast, when one studies the maximization versions of the \mvd and \umvd problems, the complete case is as hard as the general case. Indeed, as we prove in Corollary~\ref{cor:complete}, starting from an arbitrary instance, one can add edges of very high distance to make the graph complete while roughly preserving the value of the optimal solution. Furthermore, the \maxmvd problem can be encoded into a \maxumvd problem by blowing up the distances (Corollary~\ref{cor:mvd}). Our strategy to prove Theorem~\ref{thm:max-hard-complete} is therefore to prove hardness of the \maxumvd on general graphs, which is our main result in Theorem~\ref{thm:max-hard}. Here, the idea is to interpret the problem as a constraint satisfaction problem: an ultrametric distance can be interpreted as the distance between the leaves of a tree, and thus the \maxumvd problem can be thought of as assigning one such leaf to each vertex so that the distances induced by the tree match the input distances as well as possible. In order to do so, we employ standard techniques for proving Unique Games-hardness of Max-CSPs (based on dictatorship tests and invariance principles, see for instance~\cite{khot2007optimal,raghavendra2008optimal}) but the soundness of the dictatorship test requires a specific analysis that strongly relies on the tree structure of the alphabet.



\paragraph{Roadmap.}
Section~\ref{sec:mvd} is dedicated to the proof of our $O(\log n)$-approximation algorithm for our pivot-based algorithms for \mvd and \umvd, and the proof of Theorem~\ref{thm:mvd-lower}. Section~\ref{sec:ultrametricconstant} describes our $O(1)$-approximation for \umvd on unweighted complete inputs, Section~\ref{sec:umvdweighted} presents our $O(\log n \log \log n)$-approximation algorithm for weighted instances of \umvd, and in Section~\ref{sec:hardness} we prove the hardness results for the maximization versions.

\section{An $O(\log n)$-approximation for \mvd and \umvd}
\label{sec:mvd}

In this section, we prove that Algorithm~\ref{algo:mvd} provides an $O(\log n)$ approximation for the \textsc{Metric Violation Distance} problem. We then show how to modify it to also obtain an $O(\log n)$ approximation for the \umvd problem, and provide an example where any choice of pivot leads to a logarithmic approximation ratio.

We first introduce the following intuitive terminology. We view the input distance data $x$ as a family of weights on the edges of a complete graph. A \textit{triangle} is a triple of vertices, or equivalently of edges. Let $ij$ be an edge belonging to a triangle $t=ijk$. If $x(ij) > x(ik)+ x(jk)$, we say that the edge $ij$ is in \textit{excess}. If $x(ij) < |x(ik) -x(jk)|$, we say that it is in \textit{deficit}. In both cases we call the triangle $t$ \textit{unbalanced}. After pivoting at a vertex $i$, the weights of the edges adjacent to $i$ never change in the remaining recursive calls, thus we say that they are \textit{frozen}.

The intuition for the $O(\log n)$ approximation factor is the following. At each recursive call, the algorithm chooses uniformly at random a pivot vertex $i$ and repairs the metric around it in the following way: for any $j$, the distances $x(ij)$ are kept as is, and for each triangle $ijk$, the distance $x(jk)$ is modified in the minimal way so that the triangle $ijk$ is no longer unbalanced. We first prove (Lemma~\ref{L:goodstaysgood}) that such repaired triangles never become unbalanced again further down in the algorithm. Then the key difference with previous pivot-based algorithms~\cite{ailon2008aggregating,ailon2011fitting} is that throughout recursive calls, a same edge may get modified multiple times. The main idea of the algorithm is that due to the random choice of pivots, the set of distances that an edge can get assigned roughly shrinks by a half whenever it is modified, and thus (in expectation), a given edge gets modified $O(\log n)$ times. We use this to obtain a lower bound on the fractional packing for unbalanced triangles, which by linear programming duality lower bounds the optimal value of the solution.

 We start with the following easy observation which directly follows from the description of Algorithm~\ref{algo:mvd}.
 
 \begin{observation}\label{O:balance}
 At the end of a recursive call of Algorithm~\ref{algo:mvd} where $i$ was chosen as a pivot, none of the triangles adjacent to $i$ are unbalanced.
 \end{observation}

 We denote by $\mathcal{T}$ the set of all triangles, and by $\mathcal{T}'$ the set of all unbalanced triangles. We first prove that Algorithm~\ref{algo:mvd} makes progress, i.e., the set of unbalanced triangles shrinks as the algorithm progresses. The proof is a tedious but straightforward case analysis.



  \begin{lemma}\label{L:goodstaysgood}
At each execution of a pivoting step in Algorithm~\ref{algo:mvd}, no new unbalanced triangle is created.
    \end{lemma}

In particular, combined with Observation~\ref{O:balance}, Lemma~\ref{L:goodstaysgood} shows that once a pivot has been made at a vertex $i$, its adjacent triangles are repaired and never become unbalanced again.

  \begin{proof}
    Let $i,j,k,m$ be a $4$-tuple of vertices, and say that we pivot at $i$ during the algorithm. 
    Pivoting at $i$ might change the weights of the edges $jk$, $km$ and $jm$. We denote by $x$ and $x'$ the weights respectively before and after pivoting. By construction, the triangles containing $i$ are not unbalanced after the pivot. If $jkm$ is unbalanced after the pivot, if it was already unbalanced, there is nothing to prove. Otherwise, the value of at least one of its edges was modified and is now in excess or in deficit for this triangle. If one of the vertices of $jkm$ had already been chosen as a pivot, say $m$,  we argue inductively, i.e., we assume that none of the triangles adjacent to $m$ are unbalanced and prove that it is still the case with the new weights. The initialization of the induction is provided by Observation~\ref{O:balance}.

    If the weight of an edge got increased and is in deficit, then necessarily another edge got increased and is in excess for this triangle, we look at that one instead. So if $jk$ is an edge that got increased and is in excess, without loss of generality, its new value is $x'(jk)=x(ij)-x(ik)$. Note that since the weight of $jk$ has changed, neither $j$ nor $k$ had been chosen before as a pivot, as it would have frozen that edge. If neither $jm$ nor $km$ are frozen, the new weights of the edges $jm$ and $km$ satisfy $x'(jm) \geq x(ij)-x(im)$ and $x'(km)\geq x(im)-x(ik)$. If both were frozen, i.e., there had been a pivot at $m$, then the weights of the edges $jm$ and $km$ satisfy $x'(jm)=x(jm) \geq x(ij)-x(im)$ and $x'(km)=x(km)\geq x(im)-x(ik)$ because the triangles $ijm$ and $ikm$ are not unbalanced by the induction hypothesis. Therefore $x'(jk) \leq x'(jm)+x'(km)$ and $jk$ is actually not in excess.

    Likewise, if the weight of some edge got reduced and is in excess for this triangle, some other edge got reduced and is in deficit. So if $jk$ is an edge that got reduced and is in deficit, its new weight is $x'(jk)=x(ij)+x(ik)$. Again, neither $j$ not $k$ had been chosen before as a pivot. If neither $jm$ nor $km$ were frozen, the new weights of the edges $jm$ and $km$ satisfy $x'(jm) \leq x(ij)+x(im)$ and $x'(km)\geq x(im) - x(ik)$, and thus $x'(jk) \geq x'(jm) -x'(km)$. If both were frozen, i.e., there had been a pivot at $m$, then the weights of the edges $jm$ and $km$ satisfy $x'(jm) \leq x(ij)+x(im)$ and $x'(km)\geq x(im) - x(ik)$, because by the induction hypothesis, $ijm$ and $ikm$ are not unbalanced. Thus $x'(jk) \geq x'(jm) -x'(km)$. Similarly, the other inequality $x'(jk) \geq x'(km) -x'(jm)$ is satisfied, and thus $jk$ is actually not in deficit.
      \end{proof}

The following lemma shows that when pivoting at a vertex $i$, we are not only repairing the triangles adjacent to $i$, but also other triangles which required smaller changes than the ones induced by $i$.

  \begin{lemma}\label{L:betterlemma}
    \begin{enumerate}
          \item Let $e=uv$ denote an edge adjacent to $k$ unbalanced triangles for which $e$ is in excess. Order the third vertices of these triangles from $1$ to $k$ so that $1$ is the vertex inducing the biggest decrease of the weight of $e$ if it is chosen as a pivot. Then after pivoting at $i$, all the triangles $uvj$ for $j \geq i$ are no longer unbalanced.
    \item Let $e=uv$ denote an edge adjacent to $k$ unbalanced triangles for which $e$ is deficit. Order the third vertices of these triangless from $1$ to $k$ so that $1$ is the vertex inducing the biggest increase of the weight of $e$ if it is chosen as a pivot. Then after pivoting at $i$, all the triangles $uvj$ for $j \geq i$ are no longer unbalanced.

    \end{enumerate}
  \end{lemma}

This lemma is less trivial than it might appear as the weights of $uj$ and $vj$ may have also changed with the pivoting at $i$.

  \begin{proof}
Throughout the proof, we denote by $x$ the edge-weights before pivoting, and by $x'$ the edge-weights after pivoting. Note that for both items, none of the edges $uv$, $uj$ or $vj$ for $j \in [1,k]$ are frozen since they belong to unbalanced triangles and thus this would contradict Observation~\ref{O:balance} and Lemma~\ref{L:goodstaysgood}.
    
\textbf{First item.} Let $j$ be an integer such that $j>i$. After pivoting at $i$, $x'(e)=x(iu)+x(iv)$, so we want to prove that $|x'(uj)-x'(vj)|\leq x(iu)+x(iv)$ and $x(iu)+x(iv) \leq x'(uj)+x'(vj)$.

We start with the latter inequality. From the ordering and the fact that $j>i$, we know that $x(iu)+x(iv)\leq x(uj)+x(vj)$. Therefore, for the inequality to be violated, one of $uj$ or $vj$ must have been reduced. Say that it is $uj$, thus we have $x'(uj)=x(ui)+x(ij)$. Since we also have $x'(vj)\geq x(vi)-x(ij)$, we obtain $x'(uj)+x'(vj) \geq x(ui)+x(iv)$ and the inequality is actually not violated.

For the former inequality, let us assume without loss of generality that $x'(uj)-x'(vj)>x(ui)+x(iv)$, the other case being symmetric. Then we have $x'(uj)>x(ui)+x(iv)+x'(vj)$, and thus it cannot be that we have both $x'(uj)\leq x(ui)+x(ij)$ and $x(ij) \leq x(iv)+x'(vj)$. Therefore at least one of the two triangles $iuj$ or $ivj$ is unbalanced after pivoting at $i$, which is impossible by Observation~\ref{O:balance}.


\textbf{Second item.} Let $j$ be an integer such that $j>i$. After pivoting at $i$, up to symmetry we can assume that $x'(e)=x(iu)-x(iv)$, so we want to prove that $|x'(uj)-x'(vj)|\leq x(iu)-x(iv)$ and $x(iu)-x(iv)\leq x'(uj)+x'(vj)$.

We start with the former inequality, and first prove that $x'(uj)-x'(vj)\leq x(iu)-x(iv)$. From the ordering and the fact that $j>i$, we know that $x(iu)-x(iv)\geq |x(uj)-x(vj)|\geq x(uj)-x(vj)$. Therefore, for the inequality to be violated, the weight of $uj$ must have increased or the weight of $vj$ must have decreased. If the weight of $vj$ decreased, $x'(vj)=x(ij)+x(iv)$, but then with $x'(uj) \leq x(ij)+x(ui)$ we obtain $x'(uj)-x'(vj) \leq x(ui)-x(iv)$, a contradiction. If the weight of $uj$ increased, $x'(uj)=|x(ui)-x(ij)|$, which we combine with $x'(vj) \geq |x(ij)-x(vi)|$ to obtain that $x'(uj)-x'(vj) \leq x(ui)-x(vi)$. The proof that $x'(vj)-x'(uj)\leq x(iu)-x(iv)$ is similar: if the weight of $uj$ decreased, $x'(uj)=x(ij)+x(ui)$, which combined with $x'(vj) \leq x(ij)+x(vi)$ gives $x'(vj) -x'(uj) \leq x(vi)-x(ui)\leq x(ui)-x(vi)$. If the weight of $vj$ increased, $x'(vj)=|x(vi)-x(ij)|$, which we combine with $x'(uj) \geq |x(ui)-x(ij)|$ to obtain $x'(vj)-x'(uj) \leq x(ui)-x(vi)$.

For the latter inequality, if it is violated we have $x(ui) > x(vi)+x'(uj)+x'(vj)$. As in the first item, this means that one of the two triangles $iuj$ or $ivj$ is unbalanced after pivoting, which is impossible by Observation~\ref{O:balance}.
  \end{proof}

In order to prove Theorem~\ref{thm:mvd}, we use a lower bound on the value of the optimal solution which relies on linear programming duality. A hitting set for $\mathcal{T}'$ is a set of edges $H$ so that each triangle of $\mathcal{T}'$ contains at least one edge of $H$. A fractional packing for a set of triangles $\mathcal{T}'$ is a  set of values $(p_t)_{t\in \mathcal{T}'}$ so that for each edge $e$, we have $\sum_{t \ni e}p_t\leq 1$. The value of the fractional packing is $\sum_{t\in \mathcal{T'}}p_t$.

\begin{lemma}\label{L:LP}
Any solution of \mvd is at least the value of the smallest hitting set for $\mathcal{T}'$, which is itself lower bounded by the maximal value of a fractional packing of triangles in $\mathcal{T}'$.
\end{lemma}

\begin{proof}
Any solution of \mvd must hit at least one edge of each unbalanced triangle, since otherwise that triangle remains unbalanced and thus we do not obtain a valid metric. The linear programming relaxation of the hitting set problem dualizes to the problem of finding a maximal fractional packing. Therefore the value of the maximal fractional packing is a lower bound for the solution of \mvd.
\end{proof}

The value of the (non-fractional) smallest hitting set of $\mathcal{T}'$ might be different from \mvd: even though such a hitting set $H$ is hitting all the unbalanced triangles, there might be no way of replacing the value of the weight of an edge in $H$ without recreating an unbalanced triangle, an example is given in Figure~\ref{F:badexamples}. \vca{Is Figure~\ref{F:badexamples} an ok example here?}

We are now ready to prove Theorem~\ref{thm:mvd}. 


\begin{proof}[Proof of Theorem~\ref{thm:mvd}]
We first observe that the solution output by Algorithm~\ref{algo:mvd} indeed forms a metric. By Observation~\ref{O:balance}, at each round the algorithm repairs the triangles incident to the pivot, and by Lemma~\ref{L:goodstaysgood}, these triangles stay repaired as the algorithm progresses. Since at the end of the algorithm, every vertex has been chosen once as a pivot, all the triangles are repaired and thus the triangle inequality is satisfied everywhere.

We denote by ALG the output of the algorithm and by OPT the optimal solution and recall that $\mathcal{T}$ denotes the set of all triples and $\mathcal{T}'$ the set of all unbalanced triangles with respect to the input distances. For a triangle $ijk \in \binom{n}{3}$, let $A_{ijk}$ be the indicator of the event that one of them is chosen as a pivot and one of the edges (i.e., the edge not incident to the pivot) was modified as a result. Note that in the execution of the algorithm one edge can be modified many times but this event happens at most once for each triangle since after pivoting at a vertex its adjacent edges are frozen. Let $p_{ijk}=\mathbb{E}[A_{ijk}]$, where the expectation is taken over the algorithm's randomness. Then $\mathbb{E}[ALG] \leq \sum_{t \in \mathcal{T}}p_t$. By Lemma~\ref{L:goodstaysgood}, a triangle never becomes unbalanced in the course of the algorithm, therefore $p_t=0$ for $t \not\in \mathcal{T}'$. Thus $\mathbb{E}[ALG] \leq \sum_{t \in \mathcal{T'}}p_t$.

  To show that the algorithm is an $\alpha$-approximation, we prove that for every edge $e$, $q_e:=\sum_{t\in\mathcal{T}',t\ni e}p_t \leq \alpha$. This shows that $\frac{p_t}{\alpha}$ is a fractional packing for the unbalanced triangles, and therefore, by Lemma~\ref{L:LP}, we have \[\sum_{t\in \mathcal{T'}}\frac{p_t}{\alpha} \leq OPT,\] and thus \[\mathbb{E}[ALG]\leq \sum_{t\in \mathcal{T}'}p_t\leq \alpha OPT.\] 
  
  Hence, for the rest of the proof, we show that $q_e = O(\log n)$ for every $e$. Fix an edge $e$ that belongs to at least one unbalanced triangle. Note that, for any unbalanced triangle $t$, when the event indicated by $A_t$ happens, one of the three vertices of $t$ gets chosen uniformly at random, and the opposite edge gets modified. Therefore, we have that $\sum_{t\in\mathcal{T}',t\ni e}p_t/3= \mathbb{E}[\# \textrm{times that $e$ is modified}]$, and thus it suffices to bound the latter expectation.

  We will induct on the number $n'$ of unbalanced triangles that $e$ belongs to. By Lemma~\ref{L:goodstaysgood}, this number is non-increasing throughout the algorithm. Let $c_e(n')$ be an upper bound on the expected number of modifications when $e$ is adjacent to $n'$ unbalanced triangles. For $n_1$ and $n_2$ such that $n_1+n_2=n'$, let us denote by $E_{n_1,n_2}$ the event that a pivot is chosen among the vertices forming an unbalanced triangle with $e$, and among the $n'$ unbalanced triangles, pivoting at $n_1$ of them would induce an increase in the weight of $e$ and pivoting at $n_2$ of them would induce a decrease of the weight of $e$. 
  Note that among pairs $(n_1, n_2)$ such that $n_1 + n_2 = n'$, at most one $E_{n_1, n_2}$ holds. 
  Conditioned to $E_{n_1,n_2}$, for $1\leq k \leq n_1$ and $1\leq k'\leq n_2$, $F_{k,n_2}$ and $G_{n_1,k'}$ denote respectively the events that the pivot we choose induces the $k$th biggest increase, respectively the $k'$th biggest decrease (ties broken arbitrarily). Since pivots are chosen randomly, all these events are disjoint and happen with probability $1/n'$. By Lemma~\ref{L:betterlemma}, after the event $F_{k,n_2}$, at most $n_2+k-1$  unbalanced triangles remain. Likewise, after the event $G_{n_1,k}$, at most $n_1+k-1$ unbalanced triangles remain.

  We prove by induction that $c_e(n') \leq c \ln(n'+1)$ for some constant $c>2$ to be determined. For the base case, we use the easy bound $c_e(n') \leq n' \leq c \ln(n' + 1)$ for any $n' \leq c$. For $n' > c$,


  \begin{align*}
    c_e(n')&\leq 1+ \sum_{n''} Pr[\textrm{$n''$ bad triangles remain after pivoting}] c_e(n'')&& \\
    &\leq 1+\sum_{n_1+n_2=n'} Pr[E_{n_1,n_2}](\sum_{k=1}^{n_1}Pr[F_{k,n_2}]c_e(k-1+n_2)+\sum_{k'=1}^{n_2}Pr[G_{n_1,k'}]c_e(k'-1+n_1))&&\\
    &\leq 1+\max_{n_1+n_2=n'} (\sum_{k=1}^{n_1}Pr[F_{k,n_2}]c_e(k-1+n_2)+\sum_{k'=1}^{n_2}Pr[G_{n_1,k'}]c_e(k'-1+n_1))&&\\
    & \leq 1+ \max_{n_1+n_2=n'} (\sum_{k=1}^{n_1}\frac{1}{n'}c_e(k-1+n_2)+\sum_{k'=1}^{n_2}\frac{1}{n'}c_e(k'-1+n_1))&& \\
    & \leq 1+ \max_{n_1+n_2=n'} (\sum_{k=n_1+1}^{n'}\frac{1}{n'}c_e(k-1)+\sum_{k'=n_2+1}^{n'}\frac{1}{n'}c_e(k'-1) )&& \\
    & \leq^{(1)} 1+  (\sum_{k=\lfloor n'/2\rfloor+1}^{n'}\frac{1}{n'}c_e(k-1)+\sum_{k=\lceil n'/2\rceil+1}^{n'}\frac{1}{n'}c_e(k-1)) &&\\
    & \leq^{(2)} 1 +\frac{1}{n'}c\ln(\frac{(n'!)^2}{\lfloor n'/2 \rfloor! \lceil n'/2 \rceil !})&&\\
    & \leq^{(3)} 1 + \frac{c}{n'}(2+(2n'+1)\ln n'-2n'-(\lfloor n'/2 \rfloor +\frac{1}{2})\ln (\lfloor n'/2 \rfloor )-(\lceil n'/2 \rceil +\frac{1}{2})\ln (\lceil n'/2 \rceil )+n') &&\\
    & \leq (2c-\frac{c}{2}-\frac{c}{2}) \ln n'+ 1- (2 - \frac{\ln 2}{2} - \frac{\ln 2}{2} - 1 - g(n'))c &&\\
    & \leq^{(4)} c \ln (n'+1) && \\
  \end{align*}
  where $g(n') \rightarrow_{n\rightarrow \infty} 0$ and does not depend on $c$. $(1)$ comes from the fact that the sum is maximized when $n_1=\lfloor n'/2 \rfloor$, $(2)$ is the induction hypothesis, $(3)$ is a Stirling approximation $n^{n+\frac{1}{2}}e^{-n}\leq n!\leq en^{n+\frac{1}{2}}e^{-n}$. Since we are in the case $n' \geq c$, fixing $c$ a large enough constant so that $1 - (2 - \ln 2 - 1 - g(n'))c \leq 0$ ensures $(4)$.
\end{proof}

\ar{Reviewer 3: Some of the analysis could be presented more formally (such as the
proof that it is a (log n)-approximation. As mentioned above, why is
the max usage of an edge O(log n)?
}

\begin{figure}[h]
    \centering
    \def\svgwidth{\textwidth}
    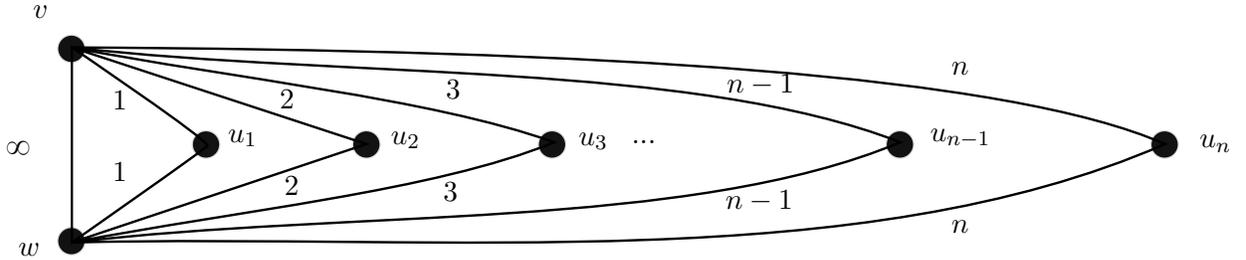
    \caption{The weights of the edges not pictured are the shortest path distances in the graph. In expectation, a solution output by Algorithm~\ref{algo:mvd} has cost $\Omega(\log n)$, while the optimal solution only changes the weight of the edge $vw$ and thus has cost $1$.}
    \label{F:lowerbound}
\end{figure}

The example in Figure~\ref{F:lowerbound} shows that the analysis in Theorem~\ref{thm:mvd} is tight.

\begin{lemma}\label{L:tight}
Let $x$ denote the distance data  described by Figure~\ref{F:lowerbound}. Then in expectation, Algorithm~\ref{algo:mvd} yields a solution of cost $\Omega(\log n)$, while the optimal solution has cost $1$.
\end{lemma}

\begin{proof}
It is straightforward to see that the optimal solution has cost $1$: it has cost at least one since there is (at least) one unbalanced triangle, and changing the weight of the $vw$ edge to a value in $[0,2]$ yields a valid metric.

For the upper bound, we first observe that the unbalanced triangles in this instance are exactly the triangles $vwu_k$. The effect of Algorithm~\ref{algo:mvd} will be to repair these triangles in some order.

If the first pivot is at a vertex $u_k$, the effect of the pivot on the weights is to change the weight of $vw$ to $2k$. After this change, the triangles $vwu_j$ for $j>k$ are no longer unbalanced, while the triangles $vwu_j$ for $j<k$ are still unbalanced. Future pivots at a vertex $u_{k'}\neq v,w$ work similarly and repair exactly the triangles $vwu_j$ for $j>k'$. 

If at any point, $v$ or $w$ is chosen as a pivot, each unbalanced triangle $vwu_k$ gets repaired by changing the weight of $vu_k$ or $wu_k$ to (roughly) $\infty$, and then the instance is fully repaired.

Therefore, the number of edges changed by Algorithm~\ref{algo:mvd} is exactly $k$, where $k$ is the smallest index of a vertex $u_k$ that was chosen as a pivot before $v$ or $w$ was chosen. Since pivots are chosen uniformly at random, for $1\leq m\leq n+2$, the probability that $v$ or $w$ gets chosen as the $m$th pivot is $\Theta(1/n)$, and the minimum index of pivots chosen until then is $\Theta(n/m)$ in expectation. Therefore the expected value output by the algorithm is $\Theta(\sum_m 1/n \times n/m)=\Theta(\log n)$.
\end{proof}

\paragraph*{\umvd and a lower bound.}We now explain how to modify Algorithm~\ref{algo:mvd} to obtain an $O(\log n)$-approximation algorithm for the \umvd problem.
We replace the triangle fixing operations so that the edge opposite to the pivot is changed minimally to satisfy the ultrametric conditions. The algorithm, originally due to Ailon and Charikar~\cite{ailon2011fitting} (designed for the $\ell_1$ objective) is given in Algorithm~\ref{algo:umvd}.\footnote{In fact, there is an approximation-preserving reduction from \umvd to \mvd, 
see the proof of Corollary~\ref{cor:mvd}. We choose to present the direct algorithm for more intuition and possible speedup. }

\begin{algorithm}[h]
    {\bf Input:} $n \in \N$, $x \in \R_{\gz}^{\nct}$, 
    {\bf Goal:} modify $x$ to an ultrametric.\\
    \If{$n \leq 2$}{{\bf return}}
        
    Choose the pivot $i \in [n]$ uniformly at random\;
    \For{$j \neq k \in [n] \setminus \{ i \}$}{
        \If{$x(i,j)=x(i,k)$}{
    $x(j, k) = \min(x(j,k),x(i,j))$;
    }
    \ElseIf{$x(j, k) \neq \max(x(i, j),x(i, k))$}{
        $x(j, k) = \max(x(i, j),x(i, k))$;
    }
    }
    \textsc{UMVD-Pivot}$(n - 1, x|_{\binom{[n]\setminus i}{2}})$ \tcp*{$x|_{\binom{[n]\setminus i}{2}}$ is $x$ restricted to $\{ (j, k) \in \binom{[n] \setminus \{ i \} }{2} \}$}
    {\bf return}
 \caption{\textsc{UMVD-Pivot}$(n, x)$}
 \label{algo:umvd}
\end{algorithm}

\begin{figure}[h]
    \centering
    \includegraphics[width=8cm]{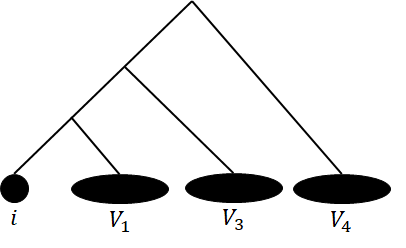}
    \caption{A pictorial description of Algorithm~\ref{algo:umvd}. Here $i$ is chosen as the pivot, and the distance from $V_j$ and $V_k$ for any $j, k$ are fixed to be $\max(j, k)$ so that every unbalanced triangle is entirely contained in some $V_j$. Then, though the algorithm syntactically considers the whole instance minus $i$, each $V_j$ does not interact with the other parts and is recursively solved.}
    \label{F:umvdpivot}
\end{figure}

Algorithm~\ref{algo:umvd} can be understood conceptually simpler than Algorithm~\ref{algo:mvd}. Consider an iteration where $i$ is chosen as the pivot. Let $x_1 < \dots < x_t$ be the distinct distances of the edges incident on $i$ (i.e., $\{ x_1, \dots, x_t \} = \{ x(i, j): j \in [n] \setminus \{ i \} \}$), and for $r \in [t]$, $V_r = \{ j \in [n] \setminus \{ i \} : x(i, j) = x_r \}$. Call each $V_r$ a {\em cluster}.  
Note that the algorithm ensures that, after the {\bf for} loop, 
$x(j, k) = \max(x(i, j), x(i, k))$ if $j$ and $k$ belong to different clusters and $x(j, k) \leq x(i, j) = x(i, k)$ if $j$ and $k$ belong to the same cluster. This implies that any triangle that is not entirely contained in one cluster becomes balanced; for any triangle $i'j'k'$ such that $i' \in V_r$ and $j', k' \in V_p$, $x(i', j') = x(i', k') = \max(x(i, i'), x(i, j')) \geq x(j', k')$. Therefore, in the subsequent iterations, different clusters do not interact at all (e.g., having a pivot in one cluster also changes the distances within the cluster), and one can recursively and separately solve each cluster, which run faster than Algorithm~\ref{algo:mvd} if the clusters are significantly smaller than $n$. (In the current implementation, one instance where the running time bound $O(n^3)$ is tight is the trivial instance where every distance value is the same.) See Figure~\ref{F:umvdpivot} for a description.

The exact same analysis technique shows that Algorithm~\ref{algo:umvd} is an $O(\log n)$-approximation algorithm for the \umvd problem: we first prove the analogues of Lemma~\ref{L:goodstaysgood} and~\ref{L:betterlemma} with a similar case analysis (which turns out to actually be simpler in this setting). Then the approximation ratio is obtained as in the proof of Theorem~\ref{thm:mvd} by bounding the expected number of times that an edge gets modified. 
For completeness, we include the proof in Appendix~\ref{appendix:umvd}.
Note that Ailon and Charikar~\cite{ailon2011fitting} first proposed this algorithm for the $\ell_1$ objective and proved it guarantees an $O(T)$-approximation where $T$ is the number of difference distance values in the instance. While $T$ can be large or small compared to $\log n$, our analysis guarantees that Algorithm~\ref{algo:umvd} achieves an $O(\min(T, \log n))$-approximation for the $\ell_0$ objective. 

The example in Figure~\ref{F:lowerbound} and Lemma~\ref{L:tight} applies equally well to the \umvd setting with no modifications. Yet that example might seem unconvincing for the reader: while a \emph{random} choice of pivots always yields a logarithmic approximation ratio, there is an obviously a much better deterministic choice of pivot: pivoting first at $u_1$ recovers the optimal solution immediately. Theorem~\ref{thm:mvd-lower} provides a more interesting bad example for \umvd, where \emph{any} choice of a sequence of pivots leads to a logarithmic approximation ratio.

\begin{proof}[Proof of Theorem~\ref{thm:mvd-lower}]We will construct an instance of \umvd with $n = 2^d$ vertices. 
Fix $d$ and consider $2^d$ vertices labeled by $\{0, 1\}^d$. 
Imagine that the hierarchy is given by the perfect binary tree with $d+1$ levels (starting from $1$) where leaves correspond to the $2^d$ vertices, 
and for each intermediate node, the two edges to the children are labeled $0$ or $1$, so that the label of each leaf is the concatenation of the labels from the root to the leaf. 

First, define all pairwise distances based on this tree. Then we add \textit{noise} on the following pairs of leaves: for each pair of $d$-bit strings that differ in one bit, we make their distance $\infty$, where $\infty$ is any integer larger than $d$. 
Then note that for each node, for each $i \in [d]$, exactly one edge originally at distance $i$ is noised to $\infty$.
The number of noised edges is $nd / 2$, which gives an upper bound on the optimal value for this noised instance. We will show that any choice of pivot is much worse than that. 

Let $T(d)$ be the optimal number of changed edges using any pivot sequence. 
We want to prove by induction that $T(d) \geq c\cdot 2^d \cdot d^2$ for some constant $c$. We will pick $c$ so that the base case is true for a constant $d$ to be fixed later. 
For the induction step, given a tree with depth $d$, by symmetry, all the pivots are the same. 
Suppose that we choose an arbitrary one. Then its $d$ {\em noised neighbors} will be fixed at distance $\infty$, 
and all edges incident to them, possibly except (1) their $\infty$ edges (2) edges between them (3) edges to the pivot, will be different than in the denoised solution. 
The number of such differences is at least $d \cdot (2^d - 2d) \geq 2^d d / 1.5$ for $d$ big enough. (This is where we fix the value of $d$.)

After such a pivot, the remaining vertices of the instance can be partitioned based on their distances to the pivot, which could be $1, \ldots d-1$ or $\infty$ and that we denote respectively by $P_1, \ldots, P_{d-1}$ and $P_\infty$. Note that after the first pivot, the weight of an edge between a vertex of $P_i$ and $P_j$ for $i\neq j$ is exactly $max(i,j)$. Furthermore, this implies that a subsequent pivot at a vertex in $P_i$ does not change any value for the weights of the edges induced by $P_j$, for $j \neq i$. Therefore, it suffices to analyze separately the future pivots for each $P_i$.

For each $i = 1, \dots, d - 1$, the set of vertices forming a $P_i$ has the metric structure of the same perfect binary tree as the original instance with depth $i$ except for one node which was sent to infinity.
In order to simplify dealing with trees without one node, we first analyze the situation when all trees are full and then subtract the appropriate correction term. Since the biggest tree has depth $d - 1$ and the imaginary node has degree at most $2^{d} / 2$, the total correction term is at most $d \cdot 2^d / 2$. 

\begin{align*}
T(d) &\geq 2^d d/1.5 + T(1) + \dots + T(d - 1) - (2^d  d/ 2) \\
&\geq 2^d d/10 + c \bigg( \sum_{t=1}^{d-1} 2^t t^2 \bigg) \\
&\geq 2^d d/10 + 2c(2^{d-1}(d-1)^2-2^{d}(d-1)+3 \cdot 2^{d-1}-3)\\
&\geq c \cdot 2^d d^2
\end{align*}
for $c$ small enough since it makes the lower order terms positive. Since the optimal solution has size at most $nd/2=2^dd/2$ and $d=\log_2(n)$, this concludes the proof.
\end{proof}

\section{Constant-factor approximation for \umvd}\label{sec:ultrametricconstant}

In this section, we provide an algorithm for the ultrametric setting and prove
Theorem~\ref{thm:umvd_noweight}. We first recall some basic properties regarding
ultrametrics. We refer to~\cite{carlsson2010characterization, Cohen-AddadKMM19} for more detailed reviews of ultrametric properties.

\paragraph{Preliminaries on ultrametrics.}
An ultrametric is a metric $(X,\dist)$, where $X$ is a set of points 
and $\dist: X \times X \mapsto \R_+$ is such that 
$\forall i,j,k \in X$, $\dist(i,j) \le \max(\dist(i,k), \dist(j,k))$.
Equivalently, an ultrametric can be represented by a rooted edge-weighted tree whose sets of leaves is $X$ and where the leaf-to-root distance is the same for all leaves, i.e.: all the elements in $X$. Then, the distance between any pair of elements is given
by the distance in the edge-weighted tree.

The tree representation $T$  of a given ultrametric is pretty useful for
algorithmic purposes. For each distance $d$, one can consider the 
\emph{clustering induced by $d$}, namely the connected components of the tree obtained from $T$ by removing the nodes at distance larger than $d/2$ from the leaves. In this case, all the leaves that are in the
same connected component are at pairwise distance at most $d$ and no
pair of leaves at distance smaller than $d$ are in different connected
components.

\paragraph{Our Algorithm.}
To simplify the exposition, we take a graph perspective on the problem: 
Given an  instance of \umvd $x \in \R_{\gz}^{\nct}$, we define a weighted graph whose vertex set is $[n]$ and the edge distances are given by the distances in $x$.
We will henceforth call the elements of $[n]$ vertices and refer to an instance of \umvd as a graph.

Throughout this section, we will use the \emph{correlation clustering} problem, 
as a subproblem to solve to obtain a constant factor approximation for the ultrametric
setting. Our algorithm defines complete graphs where each
edge is labeled with a + or a - sign. In this context, the correlation clustering problem asks
for a partition of the vertices minimizing the number of - edges that are fully contained
within a cluster plus the number of + edges across clusters. 
Our algorithm is recursive and goes as follows. 

Given an  instance of \umvd $x \in \R_{\gz}^{\nct}$, let $w_1 < \ldots < w_L$
be the distinct distance values appearing in $x$. The algorithm first creates the following
correlation clustering instance: consider the weighted graph induced by the instance $x$ as
explained above. Then, each edge of weight $w_L$ is replaced with a - edge and all the other
edges, namely of weight less than $w_L$ are + edges.
The algorithm then computes an approximate solution to the above correlation clustering instance using the algorithm 
presented in Section~\ref{sec:correlclust}, called \emph{agreement correlation clustering},
and whose properties are captured by Theorem~\ref{thm:CC} below.
The agreement correlation clustering algorithm produces a partition of the vertices. Then,
our algorithm fixes the distances between any pair  in different clusters to be $w_L$,
and creates a subinstance of the \umvd problem for each 
cluster by setting the distances between vertices
in the same clusters to be the minimum between their original distance and $w_{L-1}$.
The algorithm then makes a recursive call within each cluster.
We give a full description of the algorithm on Algorithm~\ref{algo:algo}.
\newcommand{\calA}{\mathcal{A}}
\begin{algorithm}[h]
    \caption{Procedure $\calA$: An $O(1)$-approximation for the ultrametric violation distance}
    \label{algo:algo}
    \textbf{Input:} An $n \times n$ dissimilarity matrix $x$\;
    \textbf{Output:} $\dist$: Distances between all the elements of $x$,
    such that $\dist$ is an ultrametric\;
    $w_{\max} \gets$ maximum distance entry in $x$\;
    $w_{\widecheck{\max}} \gets$ maximum distance entry in $x$ that is 
    smaller than $w_{\max}$\;
    $CC(w_{\max}) \gets$ Correlation clustering instance over the 
    elements of $x$ where each  pair $u,v$ such that the weight in $x$ is less than $w_{\max}$ is a + edge, each pair $u,v$ such that
    the weight is $w_{\max}$ is a - edge\;
    $\{C_1,\ldots,C_k\} \gets$ Solution to $CC(w_{\max})$ obtained by 
    \emph{agreement correlation clustering} (see Theorem~\ref{thm:CC})
    on $CC(w_{\max})$\;
    For any pair of $u,v$ such that $u \in C_i$, $v \in C_j$, $i \neq j$, $\dist(u,v) \gets w_{\max}$\;
    For any pair of $u,v$ such that $u,v \in C_i$, let $x(u,v) \leftarrow \min(x(u,v), w_{\widecheck{\max}})$\;
    \ForEach{$i = 1,\ldots,k$}{
        \If{$|C_i| > 1$}{
        Compute the distances between the elements of $C_i$ by making a 
        call to $\calA$ on $x$ restricted to the elements of $C_i$\;
        } \Else{
        $\dist(u,u) \gets 0$ where $u$ is the unique element of $C_i$\;
        }
    }
\end{algorithm}

The following fact follows from the definition of the algorithm and the observation
that given an \umvd instance with maximum distance $w$, the algorithm returns an ultrametric
with maximum distance $w$.
\begin{lemma}
\label{fact:ultrametric}
The above algorithm produces an ultrametric in polynomial time.
\end{lemma}
\begin{proof}
    The running time of the algorithm follows immediately from its 
    definition and Theorem~\ref{thm:CC}.
    
    We now argue that for any $i,j,k$, we have that the so-called
    "triple condition": $\dist(i,j) \le \max(\dist(i,k), \dist(k,j))$.
    Observe first that the distances defined by the algorithm are
    monotonically decreasing over the recursive calls: Namely,
    at each execution of $\calA$ where the maximum distance
    entry in $x$ is $w_{\max}$, we have that the next recursive 
    calls are such that the maximum distance is strictly smaller  than
    $w_{\max}$.
    Thus, consider the first recursive call after which $i,j,k$ are not in the same cluster produced by the agreement correlation clustering algorithm and let $w_{\max}$ be the maximum distance in $x$ at this
    recursive call. If $i,j,k$ are all in different clusters, then the distance between them is set to be $w_{\max}$ and in which case
    we have $\dist(i,j) = \dist(j,k) = \dist(i,j)$ and they satisfy
    the triple condition.
    Next, assume without loss of generality. that $i,j$ are in the same cluster and $k$ is in a 
    different cluster. In this case, $\dist(i,k) = \dist(j,k) = w_{\max}$. Moreover, as argued above, we have that 
    $i,j$ are for the first time in different clusters at following 
    recursive calls and so $\dist(i,j) < w_{\max}$. We thus have that
    $i,j,k$ satisfy the triple condition and $\dist$ is indeed
    an ultrametric.
\end{proof}

Let $\eps > 0$ be a constant satisfying
\begin{itemize}
\item $\frac{1/3 - 14\eps}{1+14\eps} > \eps/8$, and
\item $\frac{1}{1+\frac{1/3+14\eps}{2/3-14\eps}} > \eps/8$.
\end{itemize}
We can pick $\eps < 1/50$ to satisfy the above constraints.

We will use the following crucial notion. Given a correlation clustering instance, 
we say that a set of vertices $C$ is \emph{important} if for any vertex $v\in C$, 
$v$ has at most an $\eps/8$ fraction of its + neighbors outside $C$ and
is + connected to at least a $(1-\eps/8)$ fraction of the vertices in $C$.
These groups of vertices are in essence dense + regions.
Moreover, for a given set of vertices $C$, we say that $C$ is \emph{everywhere dense}
if for any $v \in C$, $v$ has a + edge to at least $2|C|/3$ vertices of $C$.
We say that a singleton is everywhere dense.

We will make use of the following theorem proved in Section~\ref{sec:correlclust}. The structure that it guarantees is illustrated in Figure~\ref{F:clusters}.
\begin{theorem}
\label{thm:CC}
  Let $\calS = \{S_1,\ldots,S_k\}$ be the set of clusters output by the agreement correlation clustering algorithm.
    Then, for any important group of vertices $C$, there is a cluster $S_i$ such that
    $C \subseteq S_i$, and $S_i$ does not intersect any other important groups of vertices disjoint from $C$.
    Moreover, any cluster $S_i \in \calS$ is everywhere dense.
\end{theorem}

\begin{figure}[t]
    \centering
    \def\svgwidth{10cm}
\begingroup%
  \makeatletter%
  \providecommand\color[2][]{%
    \errmessage{(Inkscape) Color is used for the text in Inkscape, but the package 'color.sty' is not loaded}%
    \renewcommand\color[2][]{}%
  }%
  \providecommand\transparent[1]{%
    \errmessage{(Inkscape) Transparency is used (non-zero) for the text in Inkscape, but the package 'transparent.sty' is not loaded}%
    \renewcommand\transparent[1]{}%
  }%
  \providecommand\rotatebox[2]{#2}%
  \newcommand*\fsize{\dimexpr\f@size pt\relax}%
  \newcommand*\lineheight[1]{\fontsize{\fsize}{#1\fsize}\selectfont}%
  \ifx\svgwidth\undefined%
    \setlength{\unitlength}{2199.01945778bp}%
    \ifx\svgscale\undefined%
      \relax%
    \else%
      \setlength{\unitlength}{\unitlength * \real{\svgscale}}%
    \fi%
  \else%
    \setlength{\unitlength}{\svgwidth}%
  \fi%
  \global\let\svgwidth\undefined%
  \global\let\svgscale\undefined%
  \makeatother%
  \begin{picture}(1,0.15379208)%
    \lineheight{1}%
    \setlength\tabcolsep{0pt}%
    \put(0,0){\includegraphics[width=\unitlength,page=1]{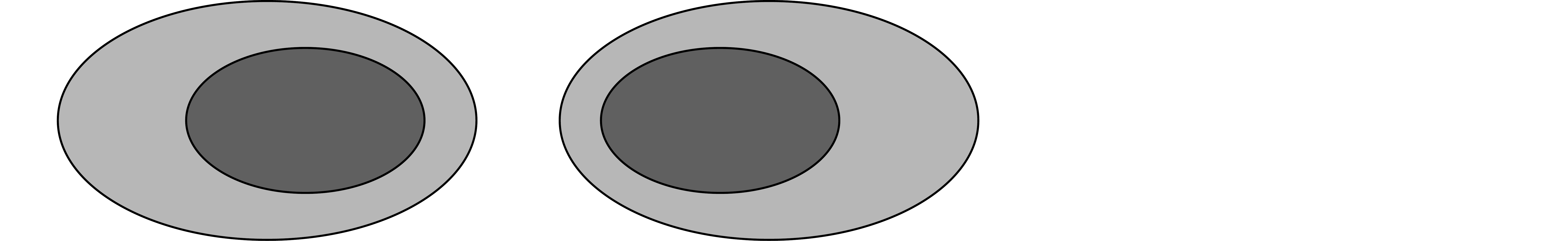}}%
    \put(-0.00113243,0.11392551){\color[rgb]{0,0,0}\makebox(0,0)[lt]{\lineheight{1.25}\smash{\begin{tabular}[t]{l}$S_1$\end{tabular}}}}%
    \put(0.62885614,0.11392551){\color[rgb]{0,0,0}\makebox(0,0)[lt]{\lineheight{1.25}\smash{\begin{tabular}[t]{l}$S_2$\end{tabular}}}}%
    \put(0.06366917,0.08396086){\color[rgb]{0,0,0}\makebox(0,0)[lt]{\lineheight{1.25}\smash{\begin{tabular}[t]{l}$C_1$\end{tabular}}}}%
    \put(0.54261639,0.08396086){\color[rgb]{0,0,0}\makebox(0,0)[lt]{\lineheight{1.25}\smash{\begin{tabular}[t]{l}$C_2$\end{tabular}}}}%
    \put(0,0){\includegraphics[width=\unitlength,page=2]{clusters.pdf}}%
    \put(0.95959891,0.11392551){\color[rgb]{0,0,0}\makebox(0,0)[lt]{\lineheight{1.25}\smash{\begin{tabular}[t]{l}$S_3$\end{tabular}}}}%
  \end{picture}%
\endgroup%

    \caption{The structure output by Algorithm~\ref{algo:algo}: the important groups of vertices (dark shaded regions) are each contained in a single cluster (lightly shaded regions) output by the algorithm. Some clusters might not contain any important group of vertices, but each cluster contains at most one of them.}
    \label{F:clusters}
\end{figure}

Assuming the above theorem we can turn
to the proof of the theorem for the ultrametric case.
\begin{proof}[Proof of Theorem~\ref{thm:umvd_noweight}] 
  Given a solution to the ultrametric problem, we define the \emph{clustering at level} $i$, to be the
  partition of the vertices obtained by putting in the same part, the vertices at pairwise distance
  less than $w_i$.
  A solution to the ultrametric problem hence defines nested clusterings. Thus, consider the following set of
  nested clusterings induced by an optimum solution to the ultrametric problem.
  The clustering $\calC_{L+1}$ at level $L+1$ simply consists of one cluster containing the whole vertex set. 
  Then, for each $w_i$, the optimal solution defines a clustering $\calC_i$ that places the vertices
  that are at distance less than $w_{i}$ in the solution in the same clusters.

  Furthermore, for each level $i$, we will consider the correlation clustering instance $M_i$ where
  for each pair $(u,v)$, there is a - edge if the input distance between $(u,v)$ is $w_i$ or higher and 
  a + edge otherwise. 

  We then apply the following top-down transformation to an optimum solution $\opt$ to the \umvd problem so as to
  obtain a solution $\opt'$ with more structure and whose cost is within a constant factor of the cost of $\opt$.
  Consider the nested clustering induced by $\opt$ as defined above and a cluster $C$ at some 
  level $i$ such 
  that either the number of - edges between pairs of vertices of $C$ in $M_i$ is at least 
  $\eps^2|C|^2/800$ or
  the cut $C, V-C$ has more than $\eps^2|C|^2/800$ + edges.
  Then, we simply set the distance between any pair of vertices of $C$ to be $w_i$
  (which implies that the clustering induced by $\opt'$ on $C$ at each level $j \le i$ consists of
  singleton elements of $C$). We will refer to this operation as \emph{making singleton clusters}
  for levels $i$ and below.
  
  Next, we proceed top-down on the non-singleton clusters of the nested clustering resulting
  from the above transformation. 
  For each cluster $C$ at level $i$, as long as there is a vertex $v$ that has
  either more than an $\eps/8$ fraction of its + neighbors outside $C$, or less than $(1-\eps/8)|C|$
  + neighbors within $C$, set the distance from $v$ to the elements of $C$ to be $w_i$.
  This implies  that in the resulting solution, $v$ is in a singleton cluster for any level $j \le i$.
  Again, we say in this case  that $v$ is \emph{made a singleton cluster} for levels $i$ and below.
  
This completes the description of $\opt'$. We have the following lemma.
\begin{lemma}
    \label{lem:costofopt}
    The cost of $\opt'$ is at most $800/\eps^2$ times the cost of $\opt$.
    Morever, at each level, the clusters induced by $\opt'$ 
    are either singletons or important groups of vertices.
\end{lemma}
\begin{proof}
  The fact that at each level, the clusters induced by $\opt'$ are
  either singletons or important groups of vertices follows immediately from
  the algorithm description.
  
  We then turn to the analysis of the cost of $\opt'$.
  The above procedure proceeds top-down through the clusters of $\opt$. We describe a charging scheme where the cost of $\opt'$ is charged to the edges {\em paid} by $\opt$ (i.e., the edges whose distances are modified in \opt). The charging occurs every time a cluster of $\opt$ is modified in $\opt'$, at which time
  each endpoint of each
  edge paid by $\opt$ will place a charge of at most
  $400/\eps^2$ on the edge. This will guarantee the $800/\eps^2$ bound in the lemma.

  Consider a cluster $C \in \opt$ at level $i$ that is split into singleton clusters
  by the procedure. If it is done at level $i$, this means that at any level 
  $i'>i$, the cluster was still intact (i.e., not split by the procedure)
  and so the edges adjacent to $C$ have not been charged by vertices of $C$ yet (note that
  they may have been charged by the endpoint not in $C$).
  
  Consider a level $i$ and the associated correlation clustering instance $M_i$.
  Assume first that $C$ has more than $\eps^2|C|^2/800$ internal - edges or
  the cut $C, V-C$ has more than $\eps^2|C|^2/800$ + edges.
  Then, we have that the cost of $\opt$ on the edges adjacent to $C$ is at least $\eps^2|C|^2/800$.
  Thus, setting the distances between vertices in $C$ to be $w_i$, or in other words
  making $C$ singleton clusters, only increases the overall cost induced by the edges
  adjacent to vertices in $C$ in $\opt$ by at most a $400/\eps^2$ factor (since there are $\binom{|C|}{2}$ edges). This can thus be paid for by placing
  a charge of $400/\eps^2$ on the edges adjacent
  to vertices in $C$ that are paid by $\opt$.
  
  Next, assume that $C$ does not satisfy the above two conditions
  but has at least one vertex that has more than an $\eps/8$ fraction of its + neighbors outside $C$ or
  that is + connected to less than a $(1-\eps/8)$ fraction of the vertices in $C$. Then, as long as there is such a vertex $v$ in $C$, the procedure removes $v$ from $C$ and repeats. Assume first 
  that the total number of vertices that can thus be removed is at most $|C|\eps/16$. 
  Then, for each vertex $u$ that is removed from $C$ we have that either it has at least $(\eps/8)(1-\eps/16)|C|$ - neighbors 
  in $C$ initially, or at least $\eps/8(|N(u)|-|C|\eps/16)\geq  (\eps/8)|C|(1-\eps/16 - \eps/16)$ + neighbors outside $C$.
  We can thus place a charge of $50/\eps < 400/\eps^2$ on each edge adjacent to $u$ that is violated by $\opt$ and this accounts for the cost of making $u$ a singleton.
  
  We now argue that the total number of vertices that have been removed is at most $|C|\eps/16$.
  Indeed, assume 
  toward contradiction that there are more than $|C|\eps/16$ vertices removed from $C$ by the above process. That implies that there were at least $\eps|C|/32$ vertices initially in $C$
  with a - edge to at least  $(\eps/8 - \eps/16)|C|$ 
  vertices of $C$ or with at least a $(1-\eps/8)|C| \eps/8 - \eps|C|/16$
  + neighbors outside of $C$.
  In which case, that means that $C$ has at least $\eps^2 |C|^2 / 800$ internal - edges or that the cut $C, V-C$ has more than
  $\eps^2|C|^2 /800$ + edges, a contradiction.

  To conclude the proof observe that when a vertex becomes singleton at some level $t$ it does not
  charge anything on its adjacent edges
  at any level $t' < t$ and so the total charge placed
  on each edge paid by $\opt$ is indeed at most $800/\eps^2$ and the lemma follows.
\end{proof}

For a given level $i$, we let the \emph{important clusters} of 
$\opt'$ be the non-singleton clusters of $\opt'$ at this level.
Note that the important clusters are important groups of vertices
that are not split in $\opt'$, i.e., in the same cluster of $\opt'$. 

Furthermore, given any solution $S$, we say that a pair of vertices $(u,v)$ is
\emph{separated at level $i$} if the distance between $u,v$ in the solution is $w_i$.
This means that the nested clusterings induced by solution $S$ are such that $u,v$
are in the same cluster for any clustering of level $j > i$ and in different clusters
in the clusterings at level $j \le i$.


We now show the following lemma.
\begin{lemma}
\label{lem:approxopt}
    Consider a solution $S$ for the \umvd problem such that  
    for any level $i$, the clustering $\calC_i$ induced by the solution at level $i$, 
    satisfies the following properties,
    \begin{enumerate}
    \item Each cluster $C \in \calC_i$ is everywhere dense;
    \item For each cluster $C \in \calC_i$ there is at most one important cluster $C'$ 
    of $\opt'$ at level $i$ that intersects $C$;
    \item Each important cluster $C'$ of $\opt'$ at level $i$ is fully contained 
    in a cluster of $\calC_i$.
    \end{enumerate}
    
    Then the cost of solution $S$ is at most 4 times the cost of solution $\opt'$.
\end{lemma}
\begin{proof}
    We have that at any level, each important cluster of $\opt'$ is fully contained
    in a cluster of $S$. Thus, consider a pair $(u,v)$ and let $i$ be the first 
    level in which $u,v$ are separated in $\opt'$, i.e.: the distance between $u,v$ in
    $\opt'$ is $w_i$. Then $(u,v)$ are not in the same cluster in $\opt'$ at level $i$.
    We have that for any level $i'>i$, $u,v$ 
    are in the same important cluster, and so they are also in the same cluster for any
    level $i' > i$ in solution $S$ as well by the property of the clusters of $S$.
    Now, two events can happen: Either the distance between $u,v$ in solution $S$ is also
    $w_i$, or in other words $u,v$ are separated at level $i$ in solution $S$ and so
    the cost induced by edge $(u,v)$ is exactly the same in $S$ as in $\opt'$.
    Or, $u,v$ are separated at a level $j < i$ in solution $S$, i.e.: their distance in $S$
    is less than $w_i$. We investigate this latter case.

    Since there is no cluster of $S$ intersecting two distinct important clusters of $\opt'$,
    at least one of $u,v$ is a singleton in $\opt'$ at level $i$.
    Therefore, the only edges that are not violated in $\opt'$ but possibly violated in $S$, namely
    for which the cost is 1 in $S$ but 0 in $\opt'$, are the
    edges $(u, v)$ such that for some level $i$, both $u$ and $v$ are in the same cluster $C$ in $S$
    but at least one of them is a singleton cluster in $\opt'$. Let $E'$ be the set of such edges.

    Our next step is to show that $|E'|$ is at most $3$ times the cost of $\opt'$ by the following scheme charging $|E'|$
    to the cost of $\opt'$.
    For each vertex $u$, let $i_u$ be the highest level where $u$ is a singleton in $\opt'$ but in a non-singleton cluster $C_u$ in $S$. 

    Let $E'_u := \{ (u, v) : v \in C_u \setminus u \}|$ and $E^{\opt'}_u := \{ (u, v) : v \in C_u \setminus u \mbox{ and } w(u, v) < i_u \}$.
    Note that $E^{\opt'}_u \subseteq E'_u$, intuitively $E^{\opt'}_u$ is the set of + edges between $u$ and $C_u \setminus u$ at level $i_u$.
    Furthermore, every edge in $E^{\opt'}_u$ is violated in $\opt'$ because $\opt'$ assigns a distance at least $w_{i_u}$ to them.
    The everywhere denseness of $C_u$ at level $i_u$ implies that $|E^{\opt'}_u| \geq 2(|C_u| - 1)/3$.

    Charge $|E'_u| = (|C_u| - 1)/|E^{\opt'}_u|$ to each edge in $E^{\opt'}_u$; so each edge in $E^{\opt'}_u$ receives a charge of at most $3/2$ from $u$. 
    For the correctness of this charging scheme, note that $E' \subseteq \cup_u E'_u$; if $(u, v) \in E'$, the definition of $E'$ says that
    ``for some level $i$, both $u$ and $v$ are in the same cluster $C$ in $S$ but at least one of them (say $u$) is a singleton cluster in $\opt'$,'' which implies that $i_u \geq i$ and $(u, v) \in E'_u$. 

    Finally, note that each edge $(u, v) \in \cup_u E^{\opt'}_u$ is charged $3/2$ at most once from each endpoint, receiving the total charge of $3$. 
\end{proof}

Equipped with this, we now prove that the algorithm of the previous section produces a good solution
by showing that the solution output is 
such that  
    for any level $i$, the clustering $\calC_i$ induced by the solution at level $i$, 
    satisfies the following properties,
    \begin{enumerate}
    \item Each cluster $C \in \calC_i$ is everywhere dense;
    \item For each cluster $C \in \calC_i$ there is at most one important cluster $C'$ 
    of $\opt'$ at level $i$ that intersects $C$;
    \item Each important cluster $C'$ of $\opt'$ at level $i$ is fully contained 
    in a cluster of $\calC_i$.
    \end{enumerate}

We are now ready to conclude the proof of the theorem.
We proceed by induction on the level to show that any cluster produced by the algorithm is
everywhere dense, each important cluster of $\opt'$ is contained in some cluster produced 
by the algorithm, and each cluster output intersects at most one important cluster of $\opt'$. At the top level, these properties trivially hold since everything is a single cluster both in $\opt'$ and in $S$.

Now, consider an important cluster $C$ of $\opt'$ at level $i$, then by definition of 
$\opt'$ it is a subset
of an important cluster $C'$ of $\opt'$ at level $i+1$. By induction hypothesis, there is 
a cluster $S'$ of the algorithm at level $i+1$ containing $C'$, we thus have
$C \subseteq C' \subseteq S'$.
Since $C$ is an important
group of vertices for the entire correlation clustering instance $M_i$ and $C \subseteq S'$, it is an important
group of vertices for the subinstance of $M_i$ induced by the vertices in $S'$ and so 
by Theorem~\ref{thm:CC}, there is a cluster $S$ of the clustering output by the algorithm at level $i$
containing $C$ and does not intersect any other important cluster $C'$ of $\opt'$ (because such $C'$ also remains important in the subinstance).
Finally, Theorem~\ref{thm:CC} also ensures that each cluster is everywhere dense.
Therefore, we can apply Lemma~\ref{lem:approxopt} to $S$ and invoke Lemma~\ref{lem:costofopt} to
conclude that $S$ is a $3200/\eps^2$-approximation to the ultrametric problem.
\end{proof}

\subsection{Agreement Correlation Clustering Algorithm -- Proof of Theorem~\ref{thm:CC}}
\label{sec:correlclust}
Given a correlation clustering instance, we let $N(u)$ be the set of + neighbors of vertex $u$.
In the remaining we refer to the `+' edges as edges and to the `-' edges as no-edges.

In the following, let $\eps \in (0, 1/50)$ be a small constant -- we stress that the constants have not been optimized
and the current bound obtained is an $\alpha$-approximation
for a large $\alpha = O(1/\epsilon^2)$, it is very likely that the current approach could easily lead
to a 1000-approximation but at the expense of a more tedious proof. 
We say that vertices $u,v$ \emph{agree} if $|N(u) \Delta N(v)| \le \eps \cdot \min(|N(u)|, |N(v)|)$, following the notion of~\cite{Cohen-AddadLMNP21}.
While the notion is similar to the one in~\cite{Cohen-AddadLMNP21}, the algorithm we use here is different: Our goal is to
obtain a constant factor approximation to correlation clustering with a specific structure, namely the one described by Theorem~\ref{thm:CC}, 
while the motivation in \cite{Cohen-AddadLMNP21} was to obtain an $O(1)$-approximate solution in the massively-parallel computation model. In particular,
the algorithm of \cite{Cohen-AddadLMNP21} does not satisfy the properties of Theorem~\ref{thm:CC}.

We let $A(v)$ denote the set of vertices that agree with $v$.
We consider the following algorithm for correlation clustering.
\begin{enumerate}
\item\label{step:vertex} Pick an arbitrary vertex $v$.
\item If $|A(v) \cap N(v)| \le (1-\eps/2) \min(|N(v)|, |A(v)|)$, then make $v$ a singleton and recurse on the remaining
  vertices.
\item Set $S(v) := A(v) \cap N(v)$.
\item\label{step:pruning} As long as there is a vertex $u$ in $S(v)$ with more than a $2\eps$ fraction of its neighbors outside $S(v)$ or
  that is connected to less than a $1-2\eps$ fraction of $S(v)$, remove $u$ from $S(v)$.
\item\label{step:removal} If $|S(v)| < (1-\eps) |A(v) \cap N(v)|$ then make $v$ a singleton and recurse on the remaining vertices. \eui{(E: Do we need step 4 and 5? If $u \in S(v)$ in step 3, then $N(u)$ is $\eps$-close to $N(v)$ (since $u \in A(v)$), which is $\eps$-close to $S(v)$ (by step 2), so $N(u)$ should be $O(\eps)$-close to $S(v)$?)}
\vca{Perhaps we don't need to indeed. I guess the risk is that if we have a vertex that is not in an important group of vertices, it may be in agreement with another vertex and attract it?}
\item\label{step:aggregating} As long as there is a vertex $u$ not in $S(v)$ with more than $1-4\eps$ fraction of its neighbors inside $S(v)$ and that is connected to at least a $1-4\eps$ fraction of $S(v)$, add $u$ to $S(v)$.
\item\label{step:final} 
If $|S(v)| > (1+3\eps) |A(v) \cap N(v)|$, make $v$ a singleton and recurse on the remaining
  vertices. Otherwise, create
 cluster $S(v)$ and recurse on the remaining vertices.
\end{enumerate}

We start our analysis with the following claim, which shows that any cluster
output by the procedure is indeed everywhere dense. 
\newcommand{\fractneighb}{\delta}
In the remaining, we let $\fractneighb := 14$.
\begin{claim}
\label{claim:basic}
Let $\eps < 1/2$ and $\delta = 14$.
For any non-singleton cluster $S(v)$ output by the algorithm, we have that
each vertex is + connected to a fraction of at least $(1-\fractneighb\eps)$ vertices of $S(v)$ 
and has at most $\fractneighb \eps|S(v)|$ + neighbors not in $S(v)$.
\end{claim}
\begin{proof}
At the end of Step~\ref{step:removal}, the vertices that are still in $S(v)$ are connected to at least $(1-2\eps)(1-\eps)|A(v) \cap N(v)|$ vertices of $S(v)$. The vertices that were added in step~\ref{step:aggregating}. are connected to at least $(1-4\eps)(1-\eps)|A(v)\cap N(v)|$ vertices of $S(v)$. Therefore, in a non-singleton cluster output by the algorithm, 
all the vertices are adjacent to at least $(1-4\eps)(1-\eps)|A(v) \cap N(v)|$ vertices
of $S(v)$, and so a fraction $\frac{(1-4\eps)(1-\eps)}{1+3\eps} > (1-\fractneighb\eps)$ of the vertices
of $S(v)$, where the last inequality uses the fact that 
the cluster was indeed created at Step~\ref{step:final}.

To bound the number of + neighbors of a vertex $u$ of $S(v)$, we distinguish two cases: 
Either $u$ has been added at Step~\ref{step:aggregating} and in which case the number of + neighbors outside of $S(v)$ is at most $4\eps|S(v)| $ as desired; Or $u$ has
not been removed at Step~\ref{step:pruning} meaning that it has at most $\eps|S(v)|$ + neighbors outside $S(v)$ at the end of  
Step~\ref{step:pruning}. The remaining steps can only decrease the number of + neighbors outside $S(v)$.
\end{proof}

To finish the proof of Theorem~\ref{thm:CC}, we have to show that for any important group of vertices $C$,
there is a cluster $S_i$ output by the algorithm such that
$C \subseteq S_i$, and $S_i$ does not intersect any other disjoint important groups of vertices, and
that each cluster output by the algorithm is everywhere dense. We break the proof of this
statement into two lemmas.
\begin{lemma}
\label{lem:structure}
    For any pair of important groups of vertices $C_i, C_j$, such that $C_i \cap C_j = \emptyset$, there
    is no $S \in \calS$ such that $C_j \cap S \neq \emptyset$ and
    $C_i \cap S \neq \emptyset $.
    
\end{lemma}

\begin{lemma}
\label{lem:cost}
Any important group of vertices $C$ is a subset of a cluster
  of $\calS$ and any cluster $S_i \in \calS$ is everywhere dense.
\end{lemma}

The above two lemmas immediately imply the proof of Theorem~\ref{thm:CC}. The intuition for these lemmas is that if an important group of vertices has a non-empty intersection with a cluster of $S$, then they mostly (i.e., up to a factor $(1-\Theta(\eps))$) coincide. Therefore, a cluster cannot contain two disjoint important group of vertices (Lemma~\ref{lem:structure}) and the entire important group of vertices will be added to the cluster in Step~\ref{step:aggregating} (Lemma~\ref{lem:cost}).

We next prove the first lemma.
\begin{proof}[Proof of Lemma~\ref{lem:structure}]
  Consider a cluster $S$ output by the algorithm. If $|S| = 1$ the lemma trivially holds for
  $S$.
  Thus, assume $|S| > 1$. 
  Assume towards contradiction that there are two disjoint groups 
  of important vertices $C_i$ and $C_j$ intersecting $S$.
  Without loss of generality, we have that $|S \cap C_j| \le 2|S|/3$.
   By Claim~\ref{claim:basic}, each vertex $v \in C_j \cap S$ 
is thus connected to at least $|S|(1/3-\delta\eps)$
  vertices not in $C_j$, and has total degree at most $|S|(1+\delta\eps)$. It follows that
  each such vertex $v$ has a $\frac{1/3-\delta\eps}{1+\delta\eps}$ fraction of its neighbors not
  in $C_j$ and so by definition of important groups of vertices, cannot be in $C_j$ since  $\frac{1/3-\delta\eps}{1+\delta\eps} > \eps/8$, which leads to contradiction. 
  \vca{constraint on $\eps,\delta$}

  \end{proof}
  
  We conclude by proving Lemma~\ref{lem:cost}.
  \begin{proof}[Proof of Lemma~\ref{lem:cost}]
  Let $C_j$ be an important group of vertices.
  We start with the following claim.\el{This claim additionally assumes that Step 5 doesn't do anything? (We may say that we won't even reach Step 6 if Step 5 acts..) Anyway where is Step 5 used? }
  Let $S(v)$ be a cluster created by the algorithm. We let 
  $S(v)_0 := A(v) \cap N(v)$, we let $S(v)_1$ be the set obtained after
  applying Step~\ref{step:pruning} to $S(v)_0$. 
  
  \begin{claim}
  \label{claim:importantvertices}
    Consider a set of vertices $S(v)_0$ formed at the first step of the algorithm. 
    We have that if the following two conditions hold:
    \begin{enumerate}
      \item At the end of the execution of Step~\ref{step:pruning} there exists a vertex $u \in S(v)_1$
        that is a vertex of an important group of vertices $C_j$; and
      \item $S(v)_0$ is not a singleton cluster;
    \end{enumerate}
    Then at the end of the execution of Step~\ref{step:aggregating} all the vertices of $C_j$ are in $S(v)$ and
    no vertex from another important group of vertices is added to $S(v)$.
    Moreover, $S(v)$ is output by the algorithm.
  \end{claim}
    \begin{proof}
Let $u$ be a vertex of an important group of vertices $C_j$. Then by definition of an important group of vertices, we have:

\begin{align}|N(u)\cap C_j|\geq(1-\eps/8)|N(u)|\label{eq:1}\end{align} and

\begin{align}|N(u) \cap C_j| \geq (1-\eps/8)|C_j|\label{eq:2}.\end{align}

If $u$ belongs to $S(v)_1$, we furthermore have

\begin{align}|N(u) \cap S(v)_1|\geq(1-2\eps)|N(u)|\label{eq:3}\end{align} and

\begin{align}|N(u) \cap S(v)_1| \geq (1-2\eps)|S(v)_1|.\label{eq:4}\end{align}

Therefore 

\begin{align*}|S(v)_1 \cap C_j|\geq |S(v)_1 \cap C_j \cap N(u)|&\geq& |N(u) \cap S(v)_1|-|N(u)\cap S(v)_1\cap \bar{C_j}|\\
&\geq^{\ref{eq:1},\ref{eq:4}}& (1-2\eps)|S(v)_1|-\eps/8|N(u)|\\
&\geq^{\ref{eq:3}}& (1-2\eps)|S(v)_1| -\eps|S(v)_1|/(8(1-2\eps))\\
&\geq& (1-2.5\eps)|S(v)_1|\end{align*}
since $\eps<3/8$.

Similarly, using equations~\ref{eq:1},~\ref{eq:2},~\ref{eq:3},~\ref{eq:4}, and the fact that $\eps$ is small, we obtain that $|S(v)_1 \cap C_j|\geq (1-2.5\eps) |C_j|$ and that $|\overline{S(v)_1} \cap C_j|\leq 2.5\eps |S(v)_1|$.
    

    We now consider the vertices added to $S(v)_1$ during Step~\ref{step:aggregating}. We show 
    that the total number 
    of vertices that are not from an important group of vertices
    is at most $\eps|S(v)_1|$ and that no vertex from a disjoint important group
    of vertices $C_i \neq C_j$ can be added at Step~\ref{step:aggregating}. 
    
    To show this, we prove the following invariant on the algorithm:
    At any point in the execution of Step~\ref{step:aggregating}, if the number of vertices that are not from $C_j$ that have been added to $S(v)$ is less than $\eps|S(v)_1|$ then we claim that no vertex from
    an important group of vertices $C_i \neq C_j$ can be added by the 
    algorithm. Indeed, any vertex added at this step must have more than a $(1-4\eps)$ fraction of its neighbors in $S(v)$, thus more than a $(1-5\eps)$ fraction of its neighbors in $S(v)_1$ which mostly (up to a factor $(1-\Theta(\eps))$) coincides with $C_j$ by the inequalities above. Therefore it cannot belong 
    to a disjoint important group of vertices $C_i \neq C_j$.
    \vca{I think we have the constraint $6\eps < 1/2$ here.} 

    Next, assume towards contradiction that $\eps|S(v)_1|$ vertices 
    that are not in $C_j$ have been added to $S(v)$ during Step~\ref{step:aggregating} and consider the first time it happens 
    during the execution of Step~\ref{step:aggregating}.
    This means that the vertices added before must not belong to an important group
    of vertices. Moreover, since at least $(1-2.5\eps)|S(v)_1|$ vertices of $S(v)_1$ are in $C_j$, any vertex that has been added, must be connected to at least $(1-8\eps)|S(v)_1|$ vertices of $C_j$.
    If $\eps|S(v)_1|$ such vertices have been added, it must be that there exists a vertex of $C_j$ which is adjacent to at least 
    $(1-8\eps)\eps|S(v)_1| > \eps|S(v)_1|/2$ such vertices. However, 
    by definition of important group of vertices, each vertex of $C_j$ is 
    adjacent to at most $\eps/8 |C_j|\leq \eps/(8(1-2.5\eps))|S(v)_1|<\eps|S(v)_1|/2$ vertices that are not in $C_j$, a contradiction.
    Thus, it must be that at most 
    $\eps|S(v)_1|$ vertices that are not in $C_j$ are ever added to $S(v)$
    at Step~\ref{step:aggregating}.

    Finally, each vertex of $C_j$ not in $S(v)_1$ has a $(1-\eps/8)$ of its neighbors in $C_j$ and thus a $(1-3\eps)$ fraction of its neighbors in $S(v)_1$. Moreover, since the total number of vertices not from $C_j$ added to $S(v)_1$ is at most $\eps|S(v)_1|$, all the vertices of $C_j - S(v)_1$ are added during Step~\ref{step:aggregating}.


    Therefore, since the total number of vertices of $C_j$ not in $S(v)_1$ is at most $3\eps|S(v)_1|$ the total number of vertices added is at most $4\eps|S(v)_1|$
    and so the cluster is indeed created at Step~\ref{step:final}.
    \end{proof}
    
    Equipped with the above claim, we can now conclude the proof.
    Let $S(v)$ be the cluster {\em created by} vertex $v$ (i.e., $v$ is the vertex chosen in Step 1). We will show that either $S(v)$ contains the entire important group of vertices $C_j$ or is disjoint from $C_j$. Note that $S(v)$ always contains $v$ even when it is a singleton cluster. Therefore, we consider the cases $v \in C_j$ and $v \notin C_j$ separately.
    
    First, suppose that $v \notin C_j$. We will show that if $S(v) \cap C_j \neq \emptyset$, then $S(v) \supseteq C_j$ (and no vertex from another important group of vertices by Lemma~\ref{lem:structure}).
    By Claim~\ref{claim:importantvertices}, we have that if the cluster contains a vertex of an important group $C_j$ after Step~\ref{step:pruning},
    then it contains all the vertices of $C_j$ at the end of Step 6.
    So we must show that if it does not contain an important
    vertex at the end of Step~\ref{step:pruning}, then no vertices of $C_j$ are added at Step~\ref{step:aggregating}.
    This follows from the fact that each vertex of $C_j$ is connected to at least a $(1-3\eps)$ fraction of vertices
    of $C_j$ and so if a vertex of $C_j$ joins $S(v)$ at Step~\ref{step:aggregating} it means that $S(v)$ already contained
    a vertex of $C_j$. Note that $S(v)$ can be reduced a singleton cluster $\{ v \}$ in Step 5 or 7, but the claim still holds since $v \notin C_j$. 
    
    Finally, we show that if $v \in C_j$, then $S(v) \supseteq C_j$ always. 
    We first argue that all pairs of vertices $u_1,u_2$ of $C_j$ agree. Indeed $u_1$ and $u_2$ are both adjacent to 
    a $(1-\eps/8)$ fraction of $C_j$ and has at most an $\eps/8$
    fraction of its neighbors outside $C_j$. It follows that
    $u_1$ and $u_2$ differ in at most $\eps|C_j|/4 + 
    (\eps/8)(d(u_1) + d(u_2))$ neighbors where $d(u_i)$ is the 
    degree of $u_i$. This is at most $\eps|C_j|/4 + 
    \frac{1+\eps}{1-\eps}(\eps/4)\min(d(u_1), d(u_2)) \le 
    \eps \min(d(u_1), d(u_2))$ for our choice of $\eps$.
    Moreover, for a pair of vertices $u_1,u_2$ from two disjoint
    group of important vertices, we show that $u_1$ and $u_2$ do not agree.
    Assume w.lo.g. that the group of important vertices $C_{\ell}$ containing $u_2$
    is smaller than the group of important vertices $C_k$ containing $u_1$, then
    $d(u_2)$ is adjacent to a $(1-\eps/8)$ fraction of vertices
    of its group of important vertices $C_\ell$ and at most
    an $\eps/8|C_\ell|$ vertices of the important group of vertices $C_k$ of $u_1$. However, $u_1$ is adjacent to at 
    least a fraction of $(1-\eps/8)$ vertices of $C_{k}$ and so at least $(1-\eps/8)|C_k|$ vertices.
    Thus, the neighborhoods of $u_1,u_2$ differ by at least $(1-\eps/4)|C_{\ell}|$ but the degree of $u_2$ is at most $(1+\eps/8)|C_{\ell}|$, and so they do not agree

    Hence, for any vertex $v \in C_j \cap A(v) \cap N(v)$, we
    have $|A(v) \cap N(v)| \ge 
    (1-\eps/2) \min(|A(v)|, |N(v)|)$ and so will not be removed
    at Step~\ref{step:pruning} and by the above argument $A(v)$ does
    not contain any vertex of a disjoint important group of vertices.
    Since $C_j$ is an important group of vertices, we have that at the end of 
    Step~\ref{step:pruning}, the size of $S(v)$ has not changed by more than an $\eps$ fraction. Therefore $S(v)$ is kept
    until Step~\ref{step:aggregating}. Since $S(v)$ does not contain any vertex of another important group of vertices
    and at least one vertex of $C_j$,
    Claim~\ref{claim:importantvertices} ensures that all the vertices of $C_j$ will be added to $S(v)$ and no
    vertex from another important group of vertices will be added to $S(v)$.
    Moreover, Claim~\ref{claim:importantvertices} further implies that 
    $S(v)$ will be output by the algorithm.
    Therefore, for each important group of vertices $C_j$, 
    there is a cluster in $\calS$ containing all of $C_j$.
    \el{Possible future suggestion: in addition to step 4, 5, could we also possibly remove step 7? It seems like Claim 3.5 is the only place it is used, and since $S(v)$ is {\em almost important} at the end of Step 4 (or 3), we can do additions in Step 6 simultaneously and bound the number of added vertices?}
\end{proof}


\section{$O(\log n \log \log n)$-Approximation for \umvd on Weighted Instances}
\label{sec:umvdweighted}
In this section, we study the following weighted version of \umvd where the input consists if distances $x \in \R^{\binom{n}{2}}_{\gz}$ and weights 
$w \in \R^{\binom{n}{2}}_{\gz}$, and the goal is to find an ultrametric $y \in \R^{\binom{[n]}{2}}_{\gz}$ such that $\sum_{(i, j)} w(i,j)\cdot \ind(x(i, j) \neq y(i, j))$ is minimized. 
We give an $O(\log n \log \log n)$-approximation algorithm for this problem. It shows a possible qualitative difference between \mvd and \umvd for weighted instances, because there is an approximation-preserving reduction from \textsc{Length-Bounded Cut} to \mvd on weighted instances~\cite{fan2020generalized}, and the best approximation ratio for the former remains at $O(n^{2/3})$ with the matching integrality gap for the standard LP relaxation~\cite{baier2010length}. 

Our algorithm closely follows the approach of Ailon and Charikar~\cite{ailon2011fitting} for the $\ell_1$ objective version.\el{Some ultrametric prelim needed.}
Let $w_1 < \dots < w_L$ be the distinct distance values of the given instance $x$. We consider the following LP relaxation whose variables are $\{ d_{ij}^t \}_{(i, j) \in \binom{[n]}{2}, t \in [L]}$ and $d_{ij}^t = 1$ indicates $y(i,j) \geq w_t$ (i.e., $i$ and $j$ are separated at level $t$ and below). For convenience, even though they are not variables, let $d_{ij}^{L + 1} = 0$ for every $(i, j) \in \binom{[n]}{2}$. 

\begin{align*}
\mbox{minimize } & 
\sum_{(i, j) \in E} w(i,j) ( d_{ij}^{x(i,j)+1} + (1 - d_{ij}^{x(i,j)}))  \\
\mbox{s.t. } & d^{t}_{ik} \leq d^t_{ij} + d^t_{jk} && \forall t \in [L], i, j, k \in [n] \\
& d^{t+1}_{ij} \leq d^t_{ij}
&& \forall t \in [L - 1], i, j \in [n] \\
& d^{t}_{ij} \in [0, 1]
&& \forall t \in [L], i, j \in [n]
\end{align*}
Suppose that we solve the above LP and obtain the solution $\{ d^t_{ij}\}_{t, i, j}$. Our rounding algorithm is the following. Let $E = \binom{[n]}{2}$ and $E_t = \{ (i, j) \in E: w_t > x(i, j) \}$ be the set of edges that are supposed to be not separated at level $t$, and $E'_t = E \setminus E_t$ be the set of edges that are supposed to be separated at level $t$. 
 
\begin{algorithm}[h]
    \If{$t = 0$}{{\bf return;}}
    Call \textsc{Cluster-Partition$(Z, t)$} to obtain partition $Z = Z_1 \cup \dots \cup Z_m$\;
    \For{$i, j \in Z_{\ell}$, $\ell \in [m]$}
    {
        $y(i, j) = \min(y(i,j), w_{t - 1})$;
    }
    \For{$1 \leq \ell < \ell' \leq m, i \in Z_{\ell}, j \in Z_{\ell'}$}
    {
        $y(i, j) = w_t$;
    }
    \For{$\ell \in [m]$}
    {
        \textsc{Hierarchical-Cluster}$(Z_{\ell}, t - 1)$
    }
    
 \caption{\textsc{Hierarchical-Cluster}$(Z, t)$}
 \label{algo:umvd-general-main}
\end{algorithm}

\begin{algorithm}[h]
    \For{$m = 1, 2, \dots$}{
    \If(\tcp*[h]{$i, j$ far apart and should be split}){$\exists i, j \in Z: d^t_{ij} > 2/3\mbox{ and }(i,j) \in E_t'$} 
    {
    \If{$A^t_Z(i, 1/3) \leq A^t_Z/2$}
    {$c=i$;}
    \Else
    {$c=j$;}
    $(*)$ Pick $r \in [0, 1/3]$ such that $w(\delta^t_Z(c, r))\leq O(\log \log n)A^t_Z(c,r) \ln (A^t_Z / A^t_Z(c,r))$\;
    $Z_m = V^t_Z(c, r), Z = Z \setminus Z_m$\;
    }
    \Else{
        $Z_m = Z$\;
        {\bf return} $Z_1 \cup \dots \cup Z_m$;
    }
    }
 \caption{\textsc{Cluster-Partition}$(Z, t)$}
 \label{algo:umvd-general-region}
\end{algorithm}
Let $\rho$ be the value of the LP relaxation divided by $n$. The relevant definitions used in Algorithm~\ref{algo:umvd-general-region} are as follows. 

\begin{align*}
    A^t_Z &= \rho |Z| + \sum_{ij \in E_t} w(i,j)d^{t}_{ij} \\
    V^t_Z(c, r) &= \{ i \in Z : d^t_{ci} \leq r \} \\
    A^t_Z(c, r) &= \rho|V^t_Z(c, r)| + \sum_{\substack{i,j \in V^t_Z(c, r)\\ (i, j) \in E_t}} w(i,j)d^{t}_{ij}
    + \sum_{\substack{i \in V^t_Z(c, r), j \notin V^t_Z(c, r) \\ (i, j) \in E_t}} w(i,j)(r - d^{t}_{ci}) \\
    \delta^t_Z(c, r) &= \{ (i, j) \in E_t | i \in  V^t_Z(c, r), j \notin V^t_Z(c, r) \}
\end{align*}
It is clear that the above rounding algorithm yields an ultrametric, so we directly proceed to the analysis of the approximation ratio. 
 
At a fixed level $t$, it is instructive to interpret the algorithms and the definitions in the context of the famous {\em region growing} rounding algorithm applied to the graph $G_t = (V, E_t)$ and the LP solution $\{ d^t_{ij} \}$ (for numerous graph partitioning problems including \textsc{Undirected Multicut})~\cite{vazirani2013approximation}; 
ignoring terms multiplied by $\rho$, $A^t_Z$ is the total value of the (fictitious) LP and $A^t_Z(c, r)$ is the contribution to the LP value contained in the region $V^t_Z(c, r)$. A standard analysis of the region growing algorithm yields the following lemma. 

\begin{lemma} [Lemma 4 of~\cite{ailon2011fitting}]
In the line marked $(*)$ in \textsc{Cluster-Partition}$(Z, t)$, one can always find $r \in [0, 1/3]$ such that 
$w(\delta^t_Z(c, r)) \leq O(\log \log n)A^t_Z(c,r) \ln (A^t_Z / A^t_Z(c,r))$.
\label{lem:region}
\end{lemma}

We include the proof for completeness.

\begin{proof}
The quantity $A^t_Z(c,r)$, viewed as a function of $r$, is differentiable outside of a finite number of values when $V_Z^t(c,r)$ just hits a new vertex. Since $A^t_Z(c,r)$ captures the volume of a ball growing around $c$, its derivative is the size of the sphere of radius $r$, which is captured exactly by $w(\delta^t_Z(c,r))$. We will prove that there is a constant $C$ so that $w(\delta^t_Z(c, r)) \leq C \log \log n A^t_Z(c,r) \ln (A^t_Z / A^t_Z(c,r))$. If it does not hold, we have for almost all $r \in [0,1/3]$

\[\frac{dA_Z^t(c,r)}{dr} > C \log \log n A^t_Z(c,r) \ln (A^t_Z/A^t_Z(c,r)),\]
 and thus
 
 \[-\frac{d(\ln \ln A^t_Z/A^t_Z(c,r))}{dr}> C\log \log n.\]
 
 Integrating over $[0,1/3]$, we obtain that $\ln\ln(A_Z^t/A_Z^t(c,0))-\ln\ln(A_Z^t/A_Z^t(c,1/3))> C\log \log n$. But it follows from the definitions that $A_Z^t(c,0) \geq \rho$ and $A_Z^t(c,1/3)\leq A^t_Z/2$ by the choice of $c$, and thus  $\ln\ln(A_Z^t/A_Z^t(c,0))-\ln\ln(A_Z^t/A_Z^t(c,1/3))\leq \ln\ln(A_Z^t/ \rho)-\ln\ln 2$ and we get a contradiction for $C$ big enough.
\end{proof}

\begin{lemma}
\textsc{Hierarchical-Cluster}$(Z, t)$ achieves an $O(\log n \log \log n)$-approximation. 
\end{lemma}
\begin{proof}
Fix an edge $(i, j) \in E$. We first analyze the error where $y(i, j) < x(i, j)$; the event that $(i, j)$ is separated at a lower level than it is supposed to. Consider the call to  \textsc{Cluster-Partition}$(Z, t)$ with $t = x(i, j)$ and $Z \subseteq [n]$ such that $i, j \in Z$, which should exist because $y(i, j) < x(i, j)$. The fact that $i, j$ were not separated in this call, together with the design of \textsc{Cluster-Partition}$(Z, t)$, imply that $d^{x(i, j)}_{ij} \leq 2/3$. This implies that $(i, j)$ contributes at least $1/3$ to the LP objective via the $(1 - d^{x(i,j)}_{ij})$ term. Therefore, the errors $y(i, j) < x(i, j)$ can be directly charged to the LP contributions of the corresponding edges. 

It remains to analyze the errors where $y(i, j) > x(i, j)$. This is caused when some call to \textsc{Cluster-Partition}$(Z, t)$ outputs a partition that cuts edges of $E_t$, the set of edges that should not be separated at level $t$. 
Note that each part $Z_i$ of the partition, except the last part $Z_m$, is of the form $V_Z^t(c, r)$, and $\delta^t_Z(c, r)$ is the set of edges cut by removing $Z_i$ from $Z$. By Lemma ~\ref{lem:region}, its total weight is at most $O(\log \log n) \ln (A^t_Z / A^t_Z(c, r))$ times $A^t_Z(c, r)$,
the contribution of $V(c, r)$ to the LP value. Charge $w(\delta^t_Z(c, r))$ to the vertices $i \in V^t_Z(c, r)$ and the edges $(i, j) \in E_t(V(c, r)) \cup \delta^t_Z(c, r)$ proportional to their contribution to $A^t_Z(c, r)$. (Recall that each vertex contributes $\rho$, edge $(i, j) \in E_t(V(c, r))$ contributes $w(i,j)d^{t}_{ij}$, and edge 
$(i, j) \in \delta^t_Z(c, r)$ contributes $w(i,j)(r - \min(d^t_{ci}, d^t_{cj}))$.)
Every edge is charged at most $d^t_{ij} \cdot O(\log \log n) \ln (A^t_Z / A^t_Z(c, r))$, and every vertex is charged at most $\rho \cdot O(\log \log n) \ln (A^t_Z / A^t_Z(c, r))$. 

Note that each edge $(i, j)$ can be charged during \textsc{Cluster-Partition}$(Z, t)$ only if $t > x(i, j)$ and $(i, j) \in Z$ are in the same cluster. Let $s$ be the maximum of $x(i, j) + 1$ and the lowest level where $i, j$ belong to the same cluster. For $s \leq t \leq L$, let $Z_t$ be the cluster containing $i$ and $j$ at level $t$, and let $A^t_{Z_t}(Z_{t-1})$ be the value of $A^t_{Z_t}(c, r)$ when $Z_{t-1} = V^t_Z(c, r)$ during \textsc{Cluster-Partition}$(Z, t)$. (Let $Z_{s - 1}$ be the part whose deletion from $Z_s$ cuts $(i, j)$ during \textsc{Cluster-Partition}$(Z_s, s)$.)

The sum of the charges applied to $(i, j)$ is at most 
\begin{align*}
      & \sum_{t = s}^L d^t_{ij} O(\log \log n) \ln (A^t_{Z_t} / A^t_{Z_t}(Z_{t-1})) \\
    \leq \,\, & d^{x(i,j)+1}_{ij} O(\log \log n) \sum_{t = s}^L  \ln (A^t_{Z_t} / A^t_{Z_t}(Z_{t-1})) \\
    \leq \,\, & d^{x(i,j)+1}_{ij} O(\log \log n)   \ln (A^L_{Z_L} / A^{s-1}_{Z_{s-1}}) \\
    \leq \,\, & d^{x(i,j)+1}_{ij} O(\log n \log \log n), \\
\end{align*}
where the second inequality uses the fact that $A^t_{Z_t} \leq A^{t+1}_{Z_{t+1}}(Z_t)$. The sum of the charges to each vertex can be analyzed similarly so that the total charge to the vertices are at most $\sum_{i \in [n]} \rho \cdot  O(\log n \log \log n) = O(\log n \log \log n \cdot LP)$. Therefore, the total cost of the algorithm is at most $O(\log n \log \log n)$ times the initial LP value. 
\end{proof}

\section{Hardness of the Maximization Versions}\label{sec:hardness}
In this section, we study the maximization version of \umvd and \mvd. We first prove the hardness for weighted general instances of \umvd in Theorem~\ref{thm:max-hard}, and will show later on in Section~\ref{S:reductions} how to reduce the complete and unweighted instances to this case. The instance contains the set of vertices $[n]$ and the set of (not necessarily complete) pairs $(i ,j) \in \binom{n}{2}$ with distance $x(i, j)$ and weight $w(i, j)$, and the goal is to compute an ultrametric $\dist : \binom{n}{2} \to \R$ to maximize the total weight of pairs $(i, j)$ with $x(i, j) = \dist(i, j)$. We prove the following hardness theorem. 

\begin{theorem}
Assuming the Unique Games Conjecture, it is NP-hard to approximate the weighted maximization version of \umvd within any constant factor. 
\label{thm:max-hard}
\end{theorem}

The \umvd problem can be interpreted as a Maximization CSP, where each vertex in a graph must be assigned to a leaf of a tree and the edge constraints to satisfy should represent the distance between the leaves in the tree. Our proof of Theorem~\ref{thm:max-hard} follows a standard approach for Max-CSPs (see for example~\cite{khot2007optimal}) of reducing from Unique Games by designing a complete and sound dictatorship test where the queries take the form of the constraints. Our proof requires specific analysis techniques due to the tree structure of the alphabet.

\subsection{Dictatorship Test}
Let $T, R$ be fixed positive integers and let $L = 2^T$. In this section, we slightly abuse notation and identify a number in $[L]$ with its binary representation $\{ 0, 1\}^T$; for an integer $\ell \in [L]$ and $i \in [T]$, let $[\ell]_i$ be the $i$th bit of the $T$-bit binary representation of $\ell$. Consider the perfect binary tree $B^*$ that has $T+1$ levels (the root is at level $0$ and the leaves are at level $T$), where each node at level $t$ is labeled with $\{ 0, 1 \}^t$; naturally $b \in \{ 0, 1\}^t$ is the parent of $b_1, b_2 \in \{ 0, 1\}^{t+1}$ if $b$ is the prefix of $b_1$ and $b_2$. Let each edge of $B$ have distance $1/2$, and for $i, j \in [L]$, let $d_B(i, j)$ be the distance between the leaves $i$ and $j$ in $B$. Note that it is $T - t + 1$ where $t$ is the smallest number such that $[i]_t \neq [j]_t$.  

Our dictatorship test is an instance of \umvd parameterized by $T$ and $R$ such that the set of vertices is $[L]^R$ and each edge $(u, v)$ has a specified distance $x(u, v) \in [T]$ and a weight $w(u, v) \geq 0$. It may happen that one $(u, v)$ has multiple edges with different distances. We use $(u,v)_t$ to denote the edge $(u, v)$ with distance $t \in [T]$. 

Let $\Omega = [L]$. For each $t \in [T]$, let $\mu_t$ be the distribution on $\Omega \times \Omega$ defined as follows:

\begin{itemize}
    \item With probability $1/2$, sample $y$ and $z$ independently from $[L]$. 
        
    \item With probability $1/2$, sample $y, z \in [L]$ uniformly from the distribution such that $[y]_j = [z]_j$ for $j = 1, \dots, T - t$ and $[y]_{T - t + 1} \neq [z]_{T - t + 1}$. (E.g., if $t = 1$, the possible outcomes are $(u_i, v_i)$ is $(1, 2), (2, 1), (3, 4), (4, 3), \dots$.) Note that $d_B(y, z) = t$ by design. 
\end{itemize}
Note that $\mu_t$ can be also interpreted as a function from $\Omega \times \Omega \to [0, 1]$ that outputs the probability of each possible pair. Also let $\mu_t^{\otimes R}$ be the product distribution on $(\Omega \times \Omega)^R$ so that for any $y, z \in (\Omega \times \Omega)^R$, $\mu_t^{\otimes R} (y, z) 
= \prod_{i=1}^R \mu_t(y_i, z_i)$. 
For every $y, z \in (\Omega \times \Omega)^R$, the weight of $(y,z)_t$ is defined to be $\mu_t^{\otimes R}(y, z)$, which is the probability of sampling $(y, z)$ under the distribution $\mu_t^{\otimes R}$. 

A solution to a \umvd instance on $L^R$ vertices and with distances in $[T]$ consists of an underlying tree $\mathcal T$ and a map $f: [L]^R \rightarrow \mathcal T$. A \emph{dictator} solution is a map $f:[L]^R \rightarrow [L]$ such that there exists $i$ in $[R]$ such that for any $y$ in $[L]^R$, $f(y)=y_i$, where we identify $[L]$ with the perfect binary tree $B$ as explained above. We now prove that the dictatorship test is complete and sound, i.e., that dictator solutions satisfy a large portion of the weights, while any solution where no coordinate has high influence satisfies a small portion of the weight.

\paragraph{Completeness.} 
\begin{lemma}\label{lem:complete}
The total weight of the edges satisfied by a dictator solution is at least $1/2$.
\end{lemma}


\begin{proof}
By construction, for each $t \in [T]$, the sampled edge $(y, z)_t$ from $\mu_t^{\otimes R}$ satisfies that $d_B(y_i, z_i) = t$ with probability at least $1/2$. Therefore, the total weight of the edges satisfied by this solution is at least $1/2$. 
\end{proof}

\subsection{Fourier analysis preliminaries}
To analyze the soundness of the test, we use the following standard tools from Gaussian bounds for correlated functions from Mossel~\cite{Mossel10}.
We define the correlation between two correlated spaces.
\begin{definition}
Given a distribution $\mu$ on $\Omega_1 \times \Omega_2$, the correlation $\rho(\Omega_1, \Omega_2; \mu)$ is defined as
\[
\rho(\Omega_1, \Omega_2; \mu) = \sup \left\{ \mathsf{Cov}[f, g] : f : \Omega_1 \rightarrow \mathbb{R}, g : \Omega_2 \rightarrow \mathbb{R}, \mathsf{Var}[f] = \mathsf{Var}[g] = 1 \right\}.
\]
\end{definition}
In our dictatorship test, for each $t \in [T]$, each $\mu_t$ samples $y, z \in [L]$ independently with probability $1/2$. Corollary 2.18 of~\cite{wenner2013circumventing} ensures that the correlation is not too high in this case. 
For the rest of this section, let $\rho := \sqrt{1/2}$. 

\begin{lemma} [Corollary 2.18 of~\cite{wenner2013circumventing}]
For any $t \in [T]$, $\rho(\Omega, \Omega, \mu_t) \leq \sqrt{1/2}$. 
\end{lemma}

\begin{definition} [\cite{Mossel10}]
For any function $F : [L]^R \to\R$, the  {\em Efron-Stein} decomposition is given by 
\[
f(y) = \sum_{S \subseteq [R]} f_S (y)
\]
where the functions $f_S$ satisfy 
\begin{itemize}
    \item $f_S$ only depends on $y_S$, the restriction of $y$ to the coordinates of $S$.
    \item For all $S \not\subseteq S'$ and all $z_{S'}$, $\E_y[f_S(y) | y_{S'} = z_{S'}] = 0$. 
\end{itemize}
\end{definition}
Based on the Efron-Stein decomposition, we can define (low-degree) influences of a function. For a function $F : [L]^R \to \R$ and $p \geq 1$, let $\| F \|_p := \E[|F(y)|^p]^{1/p}$
\begin{definition} [\cite{Mossel10}]\label{def:influence}
For any function $F : [L]^R \to \R$,
its {\em $i$th influence} is defined as 
\[
\Inf_i := \sum_{S : i \in S} \| f_S \|_2^2.
\]
Its {\em $i$th degree-$d$ influence} is defined as 
\[
\Inf^{\leq d}_i := \sum_{S : i \in S, |S| \leq d} \| f_S \|_2^2,
\]
\end{definition}

For $a, b \in [0, 1]$ and $\sigma \in [0, 1]$, let $\Gamma_{\sigma}(a, b) := \Pr[g_1 \leq \Phi^{-1}(a), g_2 \leq \Phi^{-1}(b)]$ where $g_1, g_2$ are $\sigma$-correlated standard Gaussian variables and $\Phi$ denotes the cumulative density function of a standard Gaussian. (E.g., $\Gamma_{\sigma}(a, 1) = a$ for any $\sigma$ and $\Gamma_0(a, b) = ab$.) We crucially use the following invariance principle applied to our dictatorship test. 

\begin{theorem} [\cite{Mossel10}]
For any $\eps > 0$ there exist $d \in \N$ and $\tau > 0$ such that the following is true. Let $F, G : [L]^R \to [0, 1]$ and fix any $t \in [T]$. If $\min(\Inf^{\leq d}_i[F], \Inf^{\leq d}_i[G]) \leq \tau$ for every $i \in [R]$, \[
\E_{(y, z) \in \mu_t^{\otimes R}}[F(y)G(z)] \leq \Gamma_{\rho}(\E[F(y)], \E[F(z)]) + \eps.
\]
\label{thm:invariance}
\end{theorem}

The following concavity will be useful in our analysis. 

\begin{lemma}[\cite{lee2015hardness}]
For fixed $a$, $\Gamma_{\rho}(a, b)$ is concave in $b$.
\label{lem:gamma}
\end{lemma}

We also use the following estimate on $\Gamma_\rho$. 

\begin{lemma}[\cite{khot2007optimal}]
$\Gamma_{\rho}(a, a) \leq O(a^{1.1})$. 
\label{lem:gamma_same}
\end{lemma}

Finally, for $F$ with very small $\E[F]$, we use the following basic hypercontractivity lemma that does not incur an additive loss. 

\begin{lemma}[{\cite[Chapter~10]{o2014analysis}}]
There exists $C > 0$ that depends on $L$ such that for any $t \in [T]$ and $F : [L]^R \to [0, 1]$, 
\[
\E_{(y, z) \in \mu_t^{\otimes r}} [F(y) F(z)] \leq \| F \|_2 \|F\|_{2 - C}.
\]
\label{lem:hyper}
\end{lemma}

\subsection{Soundness of Dictatorship Test}
Given $T$ and $L = 2^T$ from the construction of the dictatorship test, let $C > 0$ be a constant from Lemma~\ref{lem:hyper}, $\delta := 1/T, \eta := (1/T)^{2(2-C)/C}, \eps := \min(\eta, \delta)/T$ and use it to get $d, \tau > 0$ from Theorem~\ref{thm:invariance}.

Fix any solution to \umvd. Since the input distances are integers between $1$ and $T$, we can assume that the underlying tree $\calT$ for the solution is a level-$(T+1)$ tree (with possibly a polynomial number of nodes) where each edge has length $1/2$ and each vertex in $[L]^R$ is assigned to a leaf of $\calT$. 

For a node $u \in \calT$, let $V(u) \subseteq [L]^R$ be the set of vertices assigned to the descendants of $u$. For a set of nodes $U \subseteq \calT$, let $V(U) := \cup_{u \in U} V(u)$. For $V' \subseteq [L]^R$, let $\nu(V') = |V'| / L^R$ and for $u \in \calT$ and $U \subseteq \calT$, let $\nu(u) := \nu(V(u)), \nu(U) := \nu(V(U))$. 

For $V' \subseteq [L]^R$, let $\Ind(V') : [L]^R \to \{ 0, 1\}$ be the indicator function of $V'$. The following lemma shows the soundness guarantee of the dictatorship test. 

\begin{lemma}\label{lem:soundness}
Unless there exists $u \in \calT$ and $i \in [R]$ such that $\Inf^{\leq d}_i[\Ind(V(u))] > \tau$, the total weight of the satisfied edges is at most $O(T^{0.9})$.
\end{lemma}
\begin{proof}
Call a node $u \in \calT$ {\em heavy} if $\nu(u) \geq \delta$ and {\em light} otherwise. Let $C(u)$ be the set of children of $u$. 
For every edge $(y, z)_t$ satisfied by the solution, there exist $u, u', u'' \in \calT$ such that $u$ is on level $T - t$ and $u', u'' \in C(u)$, and $y, z \in V(u)$, $y \in V(u'), z \notin V(u')$, $z \in V(u''), y \notin V(u'')$.
We put $(y,z)_t$ into one of the following three categories. 
\begin{enumerate}
    \item Light edges: $u$ (thus $u'$, $u''$) is light. 
    \item Light-heavy: $u$ is heavy but at least one of $u'$, $u''$ is light. 
    \item Heavy edges: $u'$, $u''$ (thus $u$) are heavy. 
\end{enumerate}
Furthermore, let $y$ be {\em responsible} for $(y, z)_t$ if $\nu(u') < \nu(u'')$ (if $\nu(u') = \nu(u'')$, break the tie with the lexicographic ordering) and let $z$ be responsible otherwise. 

We upper bound the total weight of each category. For a subset $S$ of edges, let the {\em relative weight} of the set be simply the total weight of $S$ divided by $L^R$. Then for any $y \in [L]^R$, the total relative weight of the edges incident on $y$ is $2T$, which is $2$ for each distance.

\paragraph{Light-Heavy edges.}
For each vertex $y \in [L]^R$, let $u \in \calT$ be the lowest heavy node (i.e., at the largest level) such that $y \in V(u)$. Let $\ell$ be $u$'s level. Then the only heavy-light edges for which $y$ is responsible are distance-$(T-\ell)$ edges; if $(y, z)_t$ is a light-heavy edge $y$ is responsible for, we have $t \geq (T - \ell)$ since all strict descendants of $u$ are light, and we also have $t \leq (T - \ell)$, since all ancestors of $u$ are heavy so that if $(y, z)_{t'}$ with $t' > (T - \ell)$ becomes light-heavy, $z$ must be responsible for it. 
Therefore, their total relative weight is at most $2$. Summing over all vertices, the total (non-relative) weight of heavy-light edges is at most $2$. 

\paragraph{Light-light edges.}
Call a node $u \in \calT$ {\em barely light} if it is light but a child of a heavy node. Since any light-light edge is induced by $V(u)$ for some barely light $u$, to bound the total weight of light-light edges, it suffices to bound the total weight of the edges induced by $V(u)$ for some barely light $u$. Note that $\{ V(u) \}_{u :  \mbox{barely light}}$ is a collection of disjoint subsets of $[L]^R$.

Fix a barely light $u \in \calT$ and $t \in T$. The total weight of distance-$t$ edges induced by $V(u)$ is 
$\E_{(y, z) \in \mu_t^{\otimes r}} [\Ind(V(u))(y), \Ind(V(u))(z)]$. We upper bound this quantity in two methods depending on $\nu(u)$. 
\begin{enumerate}
    \item If $\nu(u) \geq \eta$, we apply Theorem~\ref{thm:invariance} so that unless $\Ind(V(u))$ has $i \in [R]$ with $\Inf_i^{\leq d}[\Ind(V(u))] > \tau$, 
    \[
    \E_{(y, z) \in \mu_t^{\otimes r}} [\Ind(V(u))(y), \Ind(V(u))(z)] \leq \Gamma_{\rho}(\nu(u), \nu(u)) + \eps \leq O(\nu(u)^{1.1}) + \eps \leq O(\delta^{0.1} \nu(u)) + \eps,
    \]
    where the second inequality follows from Lemma~\ref{lem:gamma_same}. Note that this case can apply for at most $1/\eta$ barely light nodes. 
    
    \item If $\nu(u) < \eta$, we simply apply Lemma~\ref{lem:hyper} so that regardless of the influences, we have 
    \[
    \E_{(y, z) \in \mu_t^{\otimes r}} [\Ind(V(u))(y), \Ind(V(u))(z)] \leq 
    \| \Ind(V(u)) \|_2 \| \Ind(V(u)) \|_{2 - C}.
    \]
    
    Since $\| \Ind(V(u)) \|_p = \nu(u)^{1/p}$, it can be further upper bounded as
    \[
    \nu(u)^{1/2} \cdot \nu(u)^{1/(2-C)}
    = 
    \nu(u)^{1 + C/(2(2-C))}
    \leq \eta^{C/(2(2-C))} \nu(u). 
    \]
\end{enumerate}

Summing over all barely light vertices and $t \in [T]$, the total weight of light-light edges is at most $T \cdot (\max(\delta^{0.1}, \eta^{C/(2(2-C))}) + \eps/\eta)$. Recalling $\delta = 1/T$, $\eta = (1/T)^{2(2-C)/C}$, and $\eps = \min(\eta, \delta)/T$, it is at most $O(T^{0.9})$ unless some $u \in \calT$ has an influential coordinate. 

\paragraph{Heavy-heavy edges.}
For each heavy $u \in \calT$ and its heavy child $u' \in \calT$, let $U_{u'} = \{ u'' : u'' \mbox{ is a sibling of }u'\mbox{ and } \nu(u'') > \nu(u') \}$, so that if $(y, z)_t$ is a heavy-heavy edge with $y \in V(u')$ and $z \in V(U_{u'})$, $y$ is responsible for it. Let $t$ be $T$ minus the level of $u$. Theorem~\ref{thm:invariance} and Lemma~\ref{lem:gamma} show that unless $\Ind(V(u'))$ has an coordinate $i$ with $\Inf^{\leq d}_i[\Ind(V(u'))] > \tau$, the total weight of such edges is at most 
\[
\E_{(y, z) \in \mu_t^{\otimes r}} [\Ind(V(u'))(y), \Ind(V(U_{u'}))(z)] \leq 
\Gamma_\rho(\nu(u'), \nu(U_{u'})) + \eps. 
\]
Charge this weight to each vertex of $V(u')$, so that each vertex $V(u')$ receives a charge of relative weight $\frac{\Gamma_\rho(\nu(u'), \nu(U_{u'})) + \eps}{\nu(u')} \leq \frac{\Gamma_\rho(\nu(u'), \nu(U_{u'}))}{\nu(u')} + \eps/\delta$.
After executing this process for every $u$ and $u'$, the total weight of heavy-heavy edges is charged to vertices, and it remains to upper bound the amount of the charge of each vertex $y \in [L]^R$. 

Fix $y \in [L]^R$, and $u$ be the lowest heavy node such that $y \in V(u)$. Let $\ell$ be the level of $u$ and $u_i$ be the ancestor of $u$ at level $i$ ($1 \leq i \leq \ell$). Note that $V(U_{u_1}), \dots, V(U_{u_\ell})$ are disjoint. 
Then the relative weight of the total charge on $y$ is at most 
\[
\sum_{i=1}^{\ell} \bigg( \frac{\Gamma_\rho(\nu(u_i), \nu(U_{u_i}))}{\nu(u_i)} + \eps/\delta \bigg) 
\leq 
\sum_{i=1}^{\ell} \bigg( \frac{\Gamma_\rho(\delta, \nu(U_{u_i}))}{\delta} \bigg) + T\eps/\delta
\leq 
T \cdot  \frac{\Gamma_\rho(\delta, 1/T)}{\delta} + T\eps/\delta. 
\]
where the first inequality comes from the concavity of $\Gamma_{\rho}(\cdot, \nu(U_{u_i}))$ (for fixed $i \leq \ell$, we have $\Gamma_\rho(\delta, \nu(U_{u_i})) \geq (\delta/\nu(u_i)) \cdot \Gamma_\rho(\nu(u_i), \nu(U_{u_i}))$ and the second inequality also comes from the concavity of $\Gamma_{\rho}(\delta, \cdot)$ together with the fact that $\sum_{i=1}^{\ell} \nu(U_{u_i}) \leq 1$. 
Recalling $\delta = 1/T$, $\eps = \min(\delta, \eta) / T$ and using Lemma~\ref{lem:gamma_same}, it is further upper bounded by $O(T^{0.9})$. 

\paragraph{Combining.}
Therefore, one can conclude that unless there exists $u \in \calT$ and $i \in [R]$ such that $\Inf^{\leq d}_i[\Ind(V(u))] > \tau$, the total weight of the satisfied edges is at most $O(T^{0.9})$. 
\end{proof}

\subsection{Reduction from Unique Games}
In this subsection, we introduce the reduction from the Unique Games using the dictatorship test constructed.
We first introduce the Unique Games Conjecture~\cite{Khot02}, which is stated below.
\begin{definition} [Unique Games]
An instance $\mathcal{L} (G(V \cup W, E), [R], \left\{ \pi(v, w) \right\}_{(v, w) \in E})$ of Unique Games consists of a regular bipartite graph $G(V \cup W, E)$ and a set $[R]$ of labels. For each edge $(v, w) \in E$ there is a constraint specified by a permutation $\pi(v, w) : [R] \rightarrow [R]$. Given a labeling $l : V \cup W \rightarrow [R]$, let $\Valug(l)$ be the fraction of edges satisfied by $l$, where an edge $e = (v, w)$ is said to be satisfied if $l(v) = \pi(v, w)(l(w))$. 
Let $\Optug(\mathcal{L}) = \max_l (\Valug(l))$. 
\end{definition}
\begin{conjecture} [Unique Games Conjecture~\cite{Khot02}] 
\label{conj:ug}
For any constant $\alpha > 0$, there is $R = R(\alpha)$ such that, for a Unique Games instance $\mathcal{L}$ with label set $[R]$, it is NP-hard to distinguish between
\begin{itemize}
\item $\Optug(\mathcal{L}) \geq 1 - \alpha$. 
\item $\Optug(\mathcal{L}) \leq \alpha$. 
\end{itemize}
\end{conjecture}

Fix an integer $T > 0$ that determines $L = 2^T$, $C$ from Lemma~\ref{lem:hyper}, $\delta = 1/T, \eta = (1/T)^{2(2-C)/C}$, $\eps = \min(\delta, \eta)/T$ and $d, \tau$ from Theorem~\ref{thm:invariance}. 
Given an instance of $\mathcal{L} (G(V \cup W, E), [R], \left\{ \pi(v, w) \right\}_{(v, w) \in E})$ of Unique Games, we construct an instance an \umvd. For $y \in [L]^R$ and a permutation $\pi : [R] \rightarrow [R]$, let $y \circ \pi \in [T]^R$ be defined by $(y \circ \pi)_i = (y)_{\pi^{-1}(i)}$.
\begin{itemize}
\item The set of vertices $\calV := V \times [L]^R$.
\item For each $t \in [T]$, the distance-$t$ edges are described by the following probabilistic procedure to sample a pair of vertices, where the weight of the edge $(u,v)_t$ is exactly the probability that it is sampled. 
\begin{itemize}
\item Sample $w \in W$ uniformly at random and its neighbors $v_1, v_2$ uniformly and independently. 
\item Sample $(y, z) \sim \mu_t^{\otimes R}$. Output the pair $((v_1, y \circ \pi_{v_1, w}), (v_2, z \circ \pi_{v_2, w}))$. 
\end{itemize}
\end{itemize}

\paragraph{Completeness.} Suppose that $\Valug(l) \geq 1 - \alpha$ for some labeling $l : V \cup W \rightarrow [R]$. Consider the \umvd solution based on the perfect $(T+1)$-level binary tree $B$ where each vertex $(v, y) \in \calV$ is mapped to the leaf $y_{l(v)}$ of $B$. For any $t \in [T]$, if we sample $w, v_1, v_2$ as above, with probability $1 - 2\alpha$, $\pi(v_1, w)^{-1}(l(v_1)) = \pi(v_2, w)^{-1}(l(v_2))$ and by the dictatorship test completeness in Lemma~\ref{lem:complete}, the solution satisfies at least half of the edges given $t, w, v_1, v_2$. Therefore, the total weight of the satisfied edges is at least $T(1-2\alpha)/2$. 

\paragraph{Soundness.}
Consider any solution with the underlying level-($T+1$) tree $\calT$, such a solution maps the vertices of $\calV$ to the leaves of $\calT$. In the soundness case, we start with a Unique Games instance for which no assignment can satisfy more than an $\alpha$ portion of the constraints.

As in the dictatorship test, for any $v \in V$ and $U \subseteq \calT$, let $F_{v, U} : [L]^R \to \{ 0, 1 \}$ be the indicator function of $V(U)$ for $v$; formally, $F_{v, U}(y) = 1$ if and only if $(v, y)$ is mapped to a descendant of $U$. Then
for each $w \in W$ and $U \subseteq \calT$, let $F_{w, U} : [L]^R \to [0, 1]$ be such that 
\[
F_{w, U}(y) := \E_{(v, w) \in E} [F_{v, U}(y \circ \pi_{v, w})].
\]

Then the soundness of the dictatorship test (Lemma~\ref{lem:soundness}) shows that unless there exists $u \in \calT$ and $i \in [R]$ with $\Inf^{\leq d}_i[F_{w, u}] > \tau$, the total weight of the satisfied edges given $w$ is at most $O(T^{0.9} / n)$. Call $w$ {\em good} if $w$ has $u \in \calT$ and $i \in [R]$ with $\Inf^{\leq d}_i[F_{w, u}] > \tau$. We define a labeling of the Unique Games instance as follows. First, for any good $w$, let $l(w) = i$. Let $\beta$ be the fraction of good $w$'s. From the representation of influences in terms of Fourier coefficients (see Definition~\ref{def:influence}), 
\[
\tau < \Inf_i^{\leq d}[F_{w, u}] \leq \E_{(v, w) \in E}[\Inf_{\pi_{v, w}(i)}^{\leq d}[F_{v, u}]]
\]
and we conclude that a $\tau / 2$ fraction of neighbors $v$ of $w$ have $\Inf_{\pi_{v, w}(i)}^{\leq d}(F_{v, u}) \geq \tau / 2$. 
We choose $l(v)$ uniformly from a candidate set
\[
\left\{ i : \Inf_{i}^{\leq d}[F_{v, u}] \geq \tau / 2 \mbox{ for some } u \right\}.
\]
If $v$ has no candidate, choose $l(v)$ arbitrarily.
Since $\Inf_i^{\leq d} [F_{v, u}] \leq \| F_{v, u } \|_2^2$, for each level $t \in [T]$, there are at most $O(1/\tau)$ nodes $u$ at level $t$ such that $\Inf_i^{\leq d} [F_{v, u}] \geq \tau/2$ for some $i$, and at most $O(T/\tau)$ nodes overall. For such a node $u$, since $\sum_i \Inf_{i}^{\leq d}[F_{v, u}] \leq d$, 
there are at most $O(d/\tau)$ possible $i$'s, so the total size of the above set is $O(\frac{Td}{\tau^2})$. So this labeling strategy satisfies at least $\Omega(\frac{\beta \tau^3}{Td})$ fraction of Unique Games constraints in expectation, and thus $\beta=O(\frac{\alpha Td}{\tau^3})$. 
Taking $\alpha$ (and thus $\beta$) small enough ensures that in the soundness case, the total weight of the satisfied edges is at most $O(T\beta + T^{0.9}) \leq O(T^{0.9})$, completing the proof of Theorem~\ref{thm:max-hard}.

\subsection{Reduction to the complete unweighted case and to \mvd}\label{S:reductions}

From Theorem~\ref{thm:max-hard}, we can apply standard gadgets (see~\cite{trevisan}) to obtain that the unweighted instances are also NP-hard to approximate within any constant factor.

\begin{corollary}\label{cor:unweighted}
Assuming the Unique Games Conjecture, it is NP-hard to approximate the unweighted maximization version of \umvd within any constant factor.
\end{corollary}


Then, by adding edges of high distance, one can show that \umvd is also NP-hard to approximate within any constant factor on complete graphs.

\begin{corollary}\label{cor:complete}
Assuming the Unique Games Conjecture, it is NP-hard to approximate the unweighted maximization version of \umvd on complete graphs within any constant factor.
\end{corollary}

\begin{proof}
We reduce a hard instance $G$ with $n$ vertices of weighted \umvd, as given by Corollary~\ref{cor:unweighted}, to an instance $G'$ of \umvd on complete graphs while approximately preserving the value of an optimal solution as follows. The set of vertices $V'$ of $G'$ is identical to that of $G$. We denote by $N$ the largest distance of an edge in $G$. Then all the non-edges in $G$ are ordered arbitrarily with an integer $i \in n^2$, and are filled with an edge in $G'$ of distance $N+i$. Let \opt denote an optimal solution for \umvd on $G$, and \opt' an optimal solution for \umvd on $G'$. Then \opt yields a solution of exactly the same value on $G'$, since the new edges are not satisfied in \opt. Therefore $\opt' \geq \opt$. On the other hand, in the solution $\opt'$, the number of new edges which are satisfied is at most $n$: since their distances are all different, they induce an acyclic graph. Therefore $\opt' \leq \opt +n$. Finally, we observe that in the hard instances output by Theorem~\ref{cor:unweighted}, the graphs are connected,
and thus $\opt\geq n$ since any spanning tree yields a satisfying assignment. Therefore, any constant factor approximation to $\umvd$ in $G'$ would yield a constant factor approximation to $\umvd$ in $G$, which would contradict Corollary~\ref{cor:unweighted}.
\end{proof}

Finally, we prove that \umvd can be encoded as a special instance of \mvd, yielding the following corollary.

\begin{corollary}\label{cor:mvd}
Assuming the Unique Games Conjecture, it is NP-hard to approximate the unweighted maximization version of \mvd on complete graphs within any constant factor.
\end{corollary}

\begin{proof}
We reduce a hard instance $G$ with $n$ vertices of unweighted \umvd on complete graphs, as given by Corollary~\ref{cor:complete}, to an instance $G'$ of \mvd on complete graphs while preserving the value of the optimal solution as follows. The set of vertices of $G'$ is identical to that of $G$, but the distances are blown up: an edge of distance $t$ in $G$ is changed to have distance $n^t$ in $G'$. We claim that this encode the \umvd problem into an \mvd problem as follows. Let \opt denote an optimal solution for \umvd on $G$, and \opt' an optimal solution for \mvd on $G'$. If a cycle $C=e_1,\ldots ,e_k$ is unbalanced in \opt, then without loss of generality, $x(e_1)> max(e_2,\ldots ,e_k)$. Then $n^{x(e_1)}> n^{max(x(e_2),\ldots, x(e_k))}\geq n*n^{max(x(e_2),\ldots x(e_k)}>n^{x(e_2)}+ \ldots + n^{x(e_k)}$ and thus the new cycle is also unbalanced for \mvd. In the other direction, if $C$ is balanced in \opt, we have $x(e_1)\leq max(x(e_2),\ldots, x(e_k))$ and all the permutations thereof. This implies that $n^{x(e_1)}\leq n^{x(e_2)}+ \ldots + n^{x(e_k)}$ and thus the resulting cycle in balanced in the $G'$. Since both \mvd and \umvd are equivalent to finding a hitting set for all the unbalanced cycles (see~\cite[Lemma~5.2]{fan2018metric}), this proves that the \mvd problem in $G'$ is equivalent to the \umvd problem in $G$, completing the proof.
\end{proof}

Corollaries~\ref{cor:complete} and~\ref{cor:mvd} complete the proof of Theorem~\ref{thm:max-hard-complete}.

\paragraph*{Acknowledgements.} We are grateful to Alantha Newman for helpful discussions. The work of the first, second and fourth authors is partially supported by the French ANR project ANR-18-CE40-0004-01 (FOCAL).

\bibliographystyle{alpha}
\bibliography{main}

\appendix

\section{Performance of Algorithm~\ref{algo:umvd} for \umvd}
\label{appendix:umvd}
In this section, we give a proof of the following theorem. The proof is a simplification of the proof for Algorithm~\ref{algo:mvd} for \mvd. 
\begin{theorem}
Algorithm~\ref{algo:umvd} achieves an $O(\min(T, \log n))$-approximation for \umvd, where $T$ is the number of different distance values of the input. 
\label{thm:umvd}
\end{theorem}

We recall the high-level intuition of 
Algorithm~\ref{algo:umvd} as a hierarchical clustering algorithm. Consider an iteration where $i$ is chosen as the pivot. Let $x_1 < \dots < x_t$ be the distinct distances of the edges incident on $i$ (i.e. $\{ x_1, \dots, x_t \} = \{ x(i j): j \in [n] \setminus \{ i \} \}$), and for $r \in [t]$, $V_r = \{ j \in [n] \setminus \{ i \} : x(i j) = x_r \}$. Call each $V_r$ a {\em cluster}.  
Note that the algorithm ensures that, after the {\bf for} loop, 
$x(j k) = \max(x(i j), x(i k))$ if $j$ and $k$ belong to different clusters and $x(j k) \leq x(i j) = x(i k)$ if $j$ and $k$ belong to the same cluster. 

\begin{observation}\label{obs:umvd}
Any triangle that is not entirely contained in one cluster becomes balanced. 
Therefore, in the subsequent iterations, different clusters do not interact at all (e.g., having a pivot in one cluster only changes the distances within the cluster).
\end{observation}
\begin{proof}
Suppose that $i$ is chosen as the pivot in the current iteration. 
 For any triangle $i'j'k'$ such that $i' \in V_r$ and $j', k' \in V_p$ with $p \neq r$, $x(i' j') = x(i' k') = \max(x(i i'), x(i j')) \geq x(j' k')$. 
\end{proof}

We denote by $\mathcal{T}$ the set of all triangles, and by $\mathcal{T}'$ the set of all unbalanced triangles (i.e., a triangle such that the largest length is strictly bigger than the other two lengths). We first prove that Algorithm~\ref{algo:umvd} makes progress, i.e., the set of unbalanced triangles shrinks as the algorithm progresses.



  \begin{lemma}\label{lem:umvd}
At each execution of a pivoting step in Algorithm~\ref{algo:umvd}, no new unbalanced triangle is created.
    \end{lemma}

  \begin{proof}
    Let $i,j,k,m$ be a $4$-tuple of vertices, and say that we pivot at $i$ during the algorithm. 
    By Observation~\ref{obs:umvd}, $jkl$ can possibly become unbalanced if all $j, k, l$ belong to the same cluster $V_r$ for some $r$. (I.e., $x(i j) = x(i k) = x(i l) = x_r$ before the modification according to $i$.) In this case, the algorithm states that each of $x(j k), x(k l), x(l j)$ takes the minimum of its own value and $x_r$. If $jkl$ was balanced before the modification, it means that there are at two edges that had the same largest distance. The modification preserves this property. 
      \end{proof}

We are now ready to prove Theorem~\ref{thm:umvd}. 


\begin{proof}[Proof of Theorem~\ref{thm:umvd}]
We first observe that the solution output by Algorithm~\ref{algo:umvd} indeed forms an ultrametric. By Observation~\ref{obs:umvd}, at each round the algorithm repairs the triangles incident to the pivot, and by Lemma~\ref{lem:umvd}, these triangles stay repaired as the algorithm progresses. Since at the end of the algorithm, every vertex has been chosen once as a pivot, all the triangles are repaired.

We denote by ALG the output of the algorithm and by OPT the optimal solution and recall that $\mathcal{T}$ denotes the set of all triples and $\mathcal{T}'$ the set of all unbalanced triangles with respect to the input distances. For a triangle $ijk \in \binom{n}{3}$, let $A_{ijk}$ be the indicator of the event that one of them is chosen as a pivot and one of the edges (i.e., the edge not incident to the pivot) was modified as a result. (By modification we only consider the event when the value strictly changes.) 
Note that in the execution of the algorithm one edge can be modified many times but this event happens at most once for each triangle since after pivoting at a vertex its adjacent edges are frozen. Let $p_{ijk}=\mathbb{E}[A_{ijk}]$, where the expectation is taken over the algorithm's randomness. Then $\mathbb{E}[ALG] \leq \sum_{t \in \mathcal{T}}p_t$. By Lemma~\ref{lem:umvd}, a triangle never becomes unbalanced in the course of the algorithm, therefore $p_t=0$ for $t \not\in \mathcal{T}'$. Thus $\mathbb{E}[ALG] \leq \sum_{t \in \mathcal{T'}}p_t$.

  To show that the algorithm is an $\alpha$-approximation, we prove that for every edge $e$, $q_e:=\sum_{t\in\mathcal{T}',t\ni e}p_t \leq \alpha$. This shows that $\frac{p_t}{\alpha}$ is a fractional packing for the unbalanced triangles, and therefore, we have \[\sum_{t\in \mathcal{T'}}\frac{p_t}{\alpha} \leq OPT,\] and thus \[\mathbb{E}[ALG]\leq \sum_{t\in \mathcal{T}'}p_t\leq \alpha OPT.\] 
  
  Hence, for the rest of the proof, we show that $q_e = O(\log n)$ for every $e$. 
    Fix an edge $e$ that belongs to at least one unbalanced triangle. Without loss of generality, assume $e = (1, 2)$. Note that, for any unbalanced triangle $t$, when the event indicated by $A_t$ happens, one of the three vertices of $t$ gets chosen uniformly at random, and the opposite edge gets modified. Therefore, we have that $\sum_{t\in\mathcal{T}',t\ni e}p_t/3= \mathbb{E}[\# \textrm{times that $e$ is modified}]$, and thus it suffices to bound the latter expectation by $O(\log n)$.
    
    In fact, we count the expected number of {\em bad modifications}, where bad modifications are the ones that still put the two endpoints of $e$ in the same cluster. Since $e$ cannot modified further if the endpoints belong to different clusters, the number of bad modifications differs the total number of modifications by at most $1$. Note that in a bad modification, the distance value strictly decreases, so $T$ is an upper bound for the number of bad modifications. It only remains the show another upper bound $O(\log n)$. 

Assume without loss of generality that the input distance values are integers between $1$ and $n^2$; by the design of the algorithm, only those values will ever appear during the algorithm. 
Bad modifications happen when we choose $i \geq 3$ as a pivot in the current cluster with $x(1i) = x(2i) < x(12)$. 
(Let $b = x(12)$ in the current cluster.)
We call such $i$ a {\em bad pivot at distance $x(1i)$}. 
Note that a bad pivot cannot be created by the design of the algorithm.

For $a \in \{ 1, \dots, b - 1 \}$, let $n'_a$ be the number of $i$ such that $x(1 i ) =x(2i) = a$, and $n' = \sum_{a = 1}^{b-1} n'_a$. 
If a bad pivot $i$ at distance $a$ is chosen, then in the next cluster of $\{ 1,2 \}$, $x(12)$ becomes $a$ and every pairwise distance is at most $a$, so there cannot be bad pivots at distance $\geq a$. 

We bound the expected number of bad modifications by induction on the initial number of bad pivots $n'$.
Let $c_e(n')$ be an upper bound on the expected number of bad modifications when the number of bad pivots is $n'$. 
We will prove that $c_e(i) \leq 2 \ln(i + 1)$. 
For the base case, $c_e(1) \leq 1/3 \leq 2 \ln 2$. 
Also, let $s'_a = \sum_{i=1}^{a-1} n'_a$. 
For $n' \geq 2$, 
\begin{align*}
 c_e(n') &= \Pr[\mbox{bad pivot is chosen}] + \E[c_e(n'') | \mbox{bad pivot is chosen}] 
 && (n'' = \mbox {(number of remaining bad pivots)}) \\
&\leq 1 + \bigg( \sum_{a = 1}^{b-1} \frac{n'_a}{n'} \cdot c_e(s'_a) \bigg) && \\
&\leq 1 +  \frac{1}{n'} \cdot \bigg( \sum_{i = 1}^{n'} c_e(i - 1) \bigg)&& \\
&\leq 1 + \frac{1}{n'} \cdot \int_{i=1}^{n'} (2 \ln i )di  &&\\
&= 1 + \frac{1}{n'} \cdot \bigg(2 n' \ln n' - 2n' + 2 \bigg) \leq 2 \ln (n' + 1), &&
\end{align*}
finishing the proof.
\end{proof}

\end{document}